\documentclass[a4paper,twocolumn,10pt,accepted=2019-09-23]{quantumarticle}
\pdfoutput=1
\usepackage[utf8]{inputenc}
\usepackage[english]{babel}
\usepackage[T1]{fontenc}
\usepackage{amsmath}
\usepackage{hyperref}

\usepackage{tikz}
\usepackage{lipsum}

\usepackage[cm]{fullpage}
\usepackage[T1]{fontenc}
\usepackage{amssymb, amsmath, amsthm}
\makeatletter
\def\amsbb{\use@mathgroup \M@U \symAMSb}
\makeatother
\usepackage[mathscr]{euscript}
\usepackage{mathtools}
\usepackage{verbatim}
\usepackage{bbold,bbm}
\usepackage{xcolor}
\definecolor{darkred}{RGB}{200, 0, 0}
\definecolor{darkgreen}{RGB}{0, 100, 0}
\definecolor{darkblue}{RGB}{0, 0, 200}
\renewcommand{\figurename}{Fig.}
\newcommand{\nbox}[2][9]{\hspace{#1pt} \mbox{#2} \hspace{#1pt}}
\newtheorem{prop}{Proposition}[section]

\newtheorem{obs}{Observation}[section]
\newcommand{\icfo}{ICFO-Institut de Ci{\`e}ncies Fot{\`o}niques, The Barcelona Institute of Science and Technology, 08860 Castelldefels (Barcelona), Spain}
\DeclareMathOperator{\tr}{tr}

\def \diracspacing {0.7pt}
\newcommand{\bra}[1]{\langle #1 \hspace{\diracspacing} |} % bra
\newcommand{\ket}[1]{| \hspace{\diracspacing} #1 \rangle} % ket
\newcommand{\ketbra}[2]{| \hspace{\diracspacing} #1 \rangle \langle #2 \hspace{\diracspacing} |} % ketbra with different vectors
\newcommand{\ketbraq}[1]{\ketbra{#1}{#1}} % ketbra with the same vector
\newcommand{\bramatket}[3]{\langle #1 \hspace{\diracspacing} | #2 | \hspace{\diracspacing} #3 \rangle} % bra-matrix-ket with different vectors
\newcommand{\bramatketq}[2]{\bramatket{#1}{#2}{#1}} % bra-matrix-ket with the same vector
\newcommand{\tran}[0]{^\textnormal{\tiny{T}}}
\newcommand{\hc}{^{\dagger}}
\newcommand{\norm}[2][]{#1| \! #1| #2 #1| \! #1|}

\newcommand{\abs}[2][]{#1| #2 #1|}
\newcommand{\smod}[0]{\hspace{-0.15cm}\mod}
\newcommand{\cL}{\mathcal{L}}

\newcommand{\sH}{\mathscr{H}}

\newcommand{\Andef}{A^{(n)} := \sum_{a = 0}^{d - 1} \omega^{an} F_{a}}
\newcommand{\om}[2][:]{\omega #1= \exp( 2 \pi \mathbbm{i} / #2 )}
\newcommand{\Cjndef}{ C_{j}^{(n)} := \frac{ \lambda_{n} }{\sqrt{d}} \sum_{k} \omega^{njk} B_{k}^{(n)} }
\newcommand{\powerrelation}{ C_{j}^{(n)} = \big[ C_{j}^{(1)} \big]^{n} }
\newcommand{\bobobservables}{B_{k} := \omega^{ k (k + 1) } X Z^{k}}
\newcommand{\comrelation}[3]
{
B_{#1}\hc &= - \omega \{ B_{#2}, B_{#3} \}
}
\newcommand{\aliceobservables}{U_{A} A_{j} U_{A}\hc = O_{j}^{(1)} \otimes P_{1} + O_{j}^{(2)} \otimes P_{2}}
\newcommand{\bobobservablesBM}{U_{B} B_{k} U_{B}\hc = O_{k}^{(1)} \otimes Q_{1} + O_{k}^{(2)} \otimes Q_{2}}
\newcommand{\canonicalobservables}
{
\begin{align*}
O_{0}^{(1)} &= X,& O_{1}^{(1)} &= X^{2} Z,& O_{2}^{(1)} &= Z^{2},\\
O_{0}^{(2)} &= X,& O_{1}^{(2)} &= Z^{2},& O_{2}^{(2)} &= X^{2} Z
\end{align*}
}
\newcommand{\state}{U \rho_{AB} U\hc = \Phi_{A'B'} \otimes \sigma_{A''B''}}
\begin{document}
\title{Maximal nonlocality from maximal entanglement and mutually unbiased bases, and self-testing of two-qutrit quantum systems}
\author{J{\k{e}}drzej Kaniewski}
\affiliation{Center for Theoretical Physics, Polish Academy of Sciences, Al. Lotnik\'ow 32/46, 02-668 Warsaw, Poland}
\affiliation{QMATH, Department of Mathematical Sciences, University of Copenhagen, Universitetsparken 5, 2100 Copenhagen, Denmark}
\orcid{0000-0003-1133-3786}
\author{Ivan {\v{S}}upi{\'c}}
\affiliation{\icfo}
\orcid{0000-0002-0361-6631}
\author{Jordi Tura}
\affiliation{Max-Planck-Institut f{\"u}r Quantenoptik, Hans-Kopfermann-Stra{\ss}e 1, 85748 Garching, Germany}
\orcid{0000-0002-6123-1422}
\author{Flavio Baccari}
\affiliation{\icfo}
\orcid{0000-0003-3374-5968}
\author{Alexia Salavrakos}
\affiliation{\icfo}
\orcid{0000-0003-4548-4339}
%
%\linebreak
%
\author{Remigiusz Augusiak}
\affiliation{Center for Theoretical Physics, Polish Academy of Sciences, Al. Lotnik\'ow 32/46, 02-668 Warsaw, Poland}
\orcid{0000-0003-1154-6132}
%
%\date{\today}
%

\maketitle

\begin{abstract}
Bell inequalities are an important tool in device-independent quantum information processing because their violation can serve as a certificate of relevant quantum properties. 
Probably the best known example of a Bell inequality is due to Clauser, Horne, Shimony and Holt (CHSH), which is defined in the simplest scenario involving two dichotomic measurements and whose all key properties are well understood. There have been many attempts to generalise the CHSH Bell inequality to higher-dimensional quantum systems, however, for most of them the maximal quantum violation---the key quantity for most device-independent applications---remains unknown. On the other hand, the constructions for which the maximal quantum violation can be computed, do not preserve the natural property of the CHSH inequality, namely, that the maximal quantum violation is achieved by the maximally entangled state and measurements corresponding to mutually unbiased bases.
In this work we propose a novel family of Bell inequalities which exhibit precisely these properties, and whose maximal quantum violation can be computed analytically. In the simplest scenario it recovers the CHSH Bell inequality.
These inequalities involve $d$ measurements settings, each having $d$ outcomes for an arbitrary prime number $d\geq 3$. 
We then show that in the three-outcome case our Bell inequality can be used to self-test the maximally entangled state of two-qutrits and three mutually unbiased bases at each site. Yet, we demonstrate that in the case of more outcomes, their maximal violation does not allow for self-testing in the standard sense, which motivates the definition of a new weak form of self-testing. The ability to certify high-dimensional MUBs makes these inequalities attractive from the device-independent cryptography point of view.
\end{abstract}
%
%\maketitle
%
\section{Introduction}
Nonlocality of quantum mechanics, in the sense first described by Einstein, Podolsky and Rosen~\cite{einstein35a} and later formalised by Bell~\cite{bell64a}, is arguably one of its most counterintuitive features.
%
%The fact that local measurements performed on spatially-separated quantum systems may give rise to statistics which cannot be explained by a local model is important from the foundational point of view, but it also has direct practical applications. More specifically, violating a Bell inequality certifies randomness which implies security in various cryptographic scenarios~\cite{}.
%
While the original motivation for studying Bell inequalities was to rule out a classical (local-realistic) description of the system under study, we now understand that Bell nonlocality can also be used in a constructive manner. If we assume that our system is governed by quantum mechanics, we can use Bell violations to certify specific quantum properties. One can, for instance, certify the dimension~\cite{brunner08a}, the amount of entanglement present in the state~\cite{moroder13a}, or the degree of incompatibility of the measurements~\cite{cavalcanti16a, chen16a}. Bell violations are also used to certify randomness produced in the experiment, which is often applied to \emph{device-independent cryptography}, e.g.~randomness generation/expansion~\cite{colbeck06a, pironio10a, colbeck11a, vazirani12a, miller16b, bouda14a}, quantum key distribution~\cite{barrett05b, acin06a, acin07a, reichardt13a, vazirani14a, miller16b, arnonfriedman18a} or multi-party cryptography~\cite{silman11a, kaniewski16a, ribeiro18a, ribeiro16a, ribeiro18b}. In some cases the observed nonlocal correlations give us a full description of the system under study (up to well-understood equivalences). This constitutes the most complete variant of device-independent certification and goes under the name of \emph{self-testing}~\cite{mayers98a, mayers04a}. Most self-testing schemes allow us to certify states which are locally qubits~\cite{bardyn09a, mckague14a, mckague12a, yang13a, bamps15a, wang16a, supic16a}. These results were combined to give certification schemes for higher-dimensional systems~\cite{yang13a, coladangelo17a, supic18a}. However, they are still based on the violation of many two-outcome Bell inequalities; in particular, the observables they use are constructed out of single-qubit measurements. It remains a highly nontrivial and interesting problem to propose certification schemes for quantum states of higher local dimension based on a single $d$-outcome Bell inequality and exploiting genuine $d$-outcome measurements such as higher-dimensional mutually unbiased bases. The main aim of our work is to fill this gap.

%\cite{mckague16a, mckague16b, kalev17a, andersson17a, }

The certification aspect adds a new layer of complexity to nonlocality. Given a Bell inequality it is no longer sufficient to find the local and quantum values, but one should go one step further to investigate whether the observed violation can be used for certification of quantum resources. Numerous Bell inequalities have been proposed in the last 25 years (see Ref.~\cite{brunner14a} for a comprehensive review), but their certification properties are in most cases poorly understood.

The simplest and most-studied Bell inequality is due to Clauser, Horne, Shimony and Holt (CHSH)~\cite{clauser69a}. In the CHSH scenario there are two devices which have two settings and two outcomes each. It is well known that this inequality can be maximally violated by performing maximally incompatible qubit measurements on the maximally entangled state of two qubits. In fact, this is essentially the only manner of achieving the maximal violation~\cite{tsirelson87a, summers87a, popescu92a, tsirelson93a}. Several generalisations of the CHSH inequality to Bell scenarios involving $d$-outcome measurements have been proposed. However,
for most of them the maximal quantum value is not known \cite{collins02a, buhrman05a, barrett06a, son06a}. For those for which the maximal quantum value can be computed, the maximal quantum violation is achieved by the maximally entangled state, but the optimal measurements do not correspond to mutually unbiased bases \cite{devicente15a, salavrakos17a, coladangelo18a}.

% none of these constructions preserves the most attractive properties of the CHSH inequality: %(i) its maximal quantum violation can be determined analytically, and (ii) it is achieved by %the maximally entangled state of local dimension $d$ and quantum measurements forming mutually %unbiased bases. 

%\cite{collins02a, buhrman05a, barrett06a, son06a, devicente15a, salavrakos17a, coladangelo18a}

In this work we introduce a generalisation of the CHSH Bell inequality which fills this gap.
%
%for which the optimal measurements are mutually unbiased bases. 
%
To this aim we consider the Bell functional due to Buhrman and Massar (BM)~\cite{buhrman05a} in which the settings and outcomes instead of being bits come from the set $\{0, 1, \ldots, d - 1 \}$ for some integer $d$ and the winning condition is interpreted modulo $d$. While this functional seems to be quite a natural generalisation of the CHSH one, it turns out to be surprisingly hard to analyse. In particular, the quantum value is only known for $d = 3$ (and the proof is numerical). Here we define a modification of the BM inequality for $d$ being an odd prime and show that it has several desirable features. Most importantly, we can compute the quantum value by first exhibiting a sum-of-squares (SOS) decomposition of the Bell operator and then giving an explicit quantum realisation which saturates this bound. This quantum realisation uses the maximally entangled state of local dimension $d$ and rank-1 projective measurements which are pairwise mutually unbiased. On the other hand, finding the classical value of our Bell expressions turns out to be a difficult problem and we compute it only for $d=3,5,7$. Nevertheless, we conjecture that the classical and quantum values differ for any prime $d$.

Importantly, the SOS decomposition allows us to derive explicit algebraic relations that the optimal observables must satisfy. For $d = 3$ we are able to completely solve these relations, i.e.~for this inequality we obtain a complete self-testing statement. To the best of our knowledge this is the first analytical self-testing statement, which does not rely on self-testing results for two-dimensional systems (except for Ref.~\cite{coladangelo17c} which relates the self-testing problem to representations of a certain group). Note that a partial self-testing statement for the maximally entangled state  of two qutrits has recently been proven numerically for a different Bell inequality \cite{salavrakos17a}. For $d = 5$ and $d = 7$, on the other hand, the situation becomes more complicated: we show that the maximal violation can be achieved by quantum realisations which are not equivalent according to the standard definition of self-testing.

The paper is organised as follows. In Section~\ref{sec:preliminaries} we establish notation, whereas in Section~\ref{sec:modified-BM} we explicitly state the modified BM inequality and compute its quantum value. In Section~\ref{sec:quantum-realisations} we provide a partial characterisation of the quantum realisations saturating the quantum bound and derive a self-testing statement for $d=3$. We summarise our findings and discuss some resulting open questions in Section~\ref{sec:conclusions}.
\section{Preliminaries}
\label{sec:preliminaries}

This section sets up the scenario and introduces the relevant
notation and terminology from the area of Bell nonlocality.

\subsection{Measurements and observables}
\label{sec:measurements-observables}
Throughout this work we assume all the systems to be finite-dimensional. A measurement with $d$ outcomes is a collection of positive semidefinite operators $\{ F_{a} \}_{a = 0}^{d - 1}$ satisfying $\sum_{a = 0}^{d - 1} F_{a} = \mathbb{1}$. Given a measurement and an integer $n \in \amsbb{Z}$ we define
\begin{equation*}
\Andef,
\end{equation*}
where $\om{d}$. Clearly, the operator corresponding to $n = 0$ is fixed by the normalisation condition: $A^{(0)} = \mathbb{1}$. Since this is a discrete Fourier transform, the inverse transformation is given by
\begin{equation*}
F_{a} = \frac{1}{d} \sum_{n = 0}^{d - 1} \omega^{-an} A^{(n)}.
\end{equation*}
Therefore, we may think of the operators $ A^{(1)}, \ldots, A^{(d-1)} $ as an alternative description of the measurement. This representation 
turns out to be convenient for our purposes.

Since all computations in this work are performed at the level of operators $A^{(n)}$, let us state some of their properties (see Appendix~\ref{app:measurements-observables} for proofs). For arbitrary $n$ we have
\begin{equation*}
\big[ A^{(n)} \big]\hc = A^{(-n)} \nbox{and} \big[ A^{(n)} \big]\hc A^{(n)} \leq \mathbb{1}.
\end{equation*}
Moreover, it is clear that $A^{(n)} = A^{(n + d)} = A^{(n - d)}$, i.e.~there are at most $d-1$ distinct operators (because $A^{(0)}=\mathbb{1}$).
%Indeed, the observables corresponding to $n \in \{ 1, 2, \ldots, d - 1 \}$ contain enough information to reconstruct all the measurement operators, i.e.~they can be seen as an alternative description of the measurement.
This description becomes particularly simple when the original measurement is projective, i.e.~the measurement operators satisfy $F_{a} F_{b} = \delta_{ab} F_{a}$ with $\delta_{ab}$ being the Kronecker delta. Then, the entire measurement can be encoded into a single operator: it is easy to verify that $A^{(n)} = A^{n}$ for $A := A^{(1)}$. In such a case we will refer to the unitary operator $A$ as the \emph{observable} and one can check that its spectrum is contained in $\{1, \omega, \omega^{2}, \ldots, \omega^{d - 1} \}$. To ensure that an unknown measurement is projective, it suffices to check that the operator $A^{(1)}$ is unitary, i.e.~$\big[ A^{(1)} \big]\hc A^{(1)} = \mathbb{1}$.
\subsection{Bell scenario and Bell inequalities}
\label{sec:bell-operator}
In this work we consider a bipartite Bell scenario
in which two parties, traditionally named Alice and Bob, 
share some bipartite physical system, and each of them performs a number of 
measurements on their share of this system. 
We assume that both parties perform $d$ measurements
and each measurement has exactly $d$ outcomes. 
The measurement choices of Alice and Bob are labelled by $j$ and $k$, whereas the outcomes by $a$ and $b$, respectively, and we use the convention that $a, b, j, k \in \{0, 1, \ldots, d - 1\}$.
%
%where denote the Bell scenario in which Alice and Bob have $m$ setting and $d$ outcomes by \scen{m}{d}.
Denoting the probability of observing outcomes $a$ and $b$ given settings $j$ and $k$ by $P(ab | jk)$, the above experiment, termed also Bell experiment, is described by a set of probability 
distributions $\{P(ab|jk)\}$.

\iffalse
 A Bell functional is defined by a real vector $(c_{abjk})_{abjk}$ and its value on the probability distribution $P(ab | jk)$ is given by
%
\begin{equation*}
%
\beta := \sum_{abjk} c_{abjk} P(ab | jk),
%
\end{equation*}
%
where all the indices are summed over $\{0, 1, \ldots, d - 1\}$. For a given Bell functional we denote the largest values achievable by local-realistic theories, quantum mechanics and no-signalling theories by $\beta_{L}, \beta_{Q}$ and $\beta_{NS}$, respectively. 
\fi

It is natural to assume that any correlations observed in such a Bell experiment satisfy the no-signalling principle, meaning that the outcomes of one of the parties cannot depend on the measurement choice made by the other party. Mathematically, this is expressed as a set linear constraints of the form
\begin{equation}
\sum_{a}P(ab|jk)=\sum_{a}P(ab|j'k)
\end{equation}
for all $b,k$ and $j\neq j'$, and
\begin{equation}
\sum_{b}P(ab|jk)=\sum_{a}P(ab|jk')
\end{equation}
for all $a,j$ and $k\neq k'$. By the very definition, correlations obeying the no-signalling
principle form a convex polytope that for further purposes we denote $\mathcal{N}_d$. 

Imagine now that Alice and Bob share a quantum system represented by a bipartite density matrix $\rho_{AB}$ and perform quantum measurements on it. Then, the conditional probabilities are given by the following well-known formula	
\begin{equation}\label{BornForm}
P(ab | jk) = \tr \big[ ( F_{a}^{j} \otimes G_{b}^{k} ) \rho_{AB} \big],
\end{equation}
where $\{F_{a}^{j}\}, \{G_{b}^{k}\}$ represent the measurements of Alice and Bob, respectively.
Let us notice that all such quantum correlations, that is, correlations obtained from quantum states (if we do not constrain their local dimension) form a convex set, denoted $\mathcal{Q}_{d}$, whose structure is in general unknown and difficult to characterise~\cite{slofstra17a}.

The last set of correlations we need to introduce here is that of correlations 
admiting the local-realistic description, that is, those for which $P(ab|jk)$ can be represented as
\begin{equation}
P(ab|jk)=\sum_{\lambda}P(\lambda)D_A(a|x,\lambda)D_B(b|y,\lambda),
\end{equation}
where $\lambda$ are the hidden variables, while $D_A(a|x,\lambda)$ (and similarly 
$D_B(b|y,\lambda)$) is a deterministic function that for a given $x$ and 
$\lambda$ returns a fixed outcome with probability one. Correlations admitting such description
are usually referred to as local or classical, and, similarly to the no-signalling
correlations, they form a convex set that is a polytope, denoted $\mathcal{P}_d$. 

For a given scenario it holds that $\mathcal{P}_d \subsetneq \mathcal{Q}_d\subsetneq\mathcal{N}_d$. In particular, a natural way to show that $\mathcal{P}_d \subsetneq \mathcal{Q}_d$ is to use Bell inequalities. To define them explicitly let us consider a Bell functional
\begin{equation}\label{BellIneq}
\beta := \sum_{abjk} c_{abjk} P(ab | jk),
\end{equation}
where $c_{abjk}$ are some real coefficients. Computing then the maximal value of $\beta$ over the local set $\mathcal{P}_d$, that is, $\beta_C:=\max_{\mathcal{P}_d}\beta$, one arrives at a Bell inequality
\begin{equation}
\sum_{abjk} c_{abjk} P(ab | jk)\leq \beta_C
\end{equation}
whose violation indicates nonlocality.
Let us also denote by $\beta_Q$ and $\beta_{NS}$ the maximal values of
$\beta$ over the quantum and no-signalling sets, that is, 
\begin{equation}
\beta_Q:=\sup_{\mathcal{Q}_d}\beta,\qquad \beta_{NS}:=\max_{\mathcal{N}_d}\beta.
\end{equation}

\subsection{Bell operators and sum of squares decompositions}

Let us now consider a Bell experiment performed on a quantum state $\rho_{AB}$ with measurements $\{F_a^j\}$ and $\{G_b^k\}$. Due to Born's formula (\ref{BornForm}) the value of a Bell functional $\beta$ corresponding to a Bell inequality (\ref{BellIneq}) can be computed as $\beta = \tr (W \rho_{AB})$, where
\begin{equation}\label{BellOp2}
W := \sum_{abjk} c_{abjk} F_{a}^{j} \otimes G_{b}^{k}
\end{equation}
is the Bell operator constructed from the measurements $\{F_a^j\}$ and $\{G_b^k\}$. 

It will be highly beneficial for us to reformulate the Bell operator in terms of
quantum observables instead of positive semi-definite measurements operators, in particular it will facilitate devising a sum of squares decomposition for our Bell inequality. To be more precise, 
let us denote the Fourier transforms of the measurements $\{ F_{a}^{j} \}$ and $\{ G_{b}^{k} \}$ by $\{A_{j}^{(n)}\}$ and $\{B_{k}^{(n)}\}$ (see Sec. \ref{sec:measurements-observables}), respectively. This allows us to rewrite
the Bell operator (\ref{BellOp2}) in the following form
\begin{equation*}
W = \frac{1}{d^{2}} \sum_{abjk} \sum_{n_{1} n_{2}} c_{abjk} \omega^{ -a n_{1} - b n_{2} } A_{j}^{(n_{1})} \otimes B_{k}^{(n_{2})},
\end{equation*}
where as before the summations go over $\{0, 1, \ldots, d - 1\}$. The coefficients
\begin{equation*}
u_{n_{1}, n_{2}, j, k} := \frac{1}{d^{2}} \sum_{ab} c_{abjk} \omega^{ -a n_{1} - b n_{2} }
\end{equation*}
correspond to the 2-dimensional discrete Fourier transform of the Bell coefficients $(c_{abjk})$. Since $(c_{abjk})$ are real, the Fourier coefficients satisfy
\begin{equation}
\label{eq:u-condition}
u_{n_{1}, n_{2}, j, k} = u_{ d - n_{1}, d - n_{2}, j, k}^{*}.
\end{equation}
We will later use the fact that this condition is also sufficient for the Bell coefficients to be real.

Analytic bounds on the quantum value of a Bell operator can be obtained by constructing an SOS decomposition. More specifically, suppose that we can show that for all valid measurements of Alice and Bob we have
\begin{equation*}
W \leq c \, \mathbb{1} - \sum_{j = 1}^{t} L_{j}\hc L_{j},
\end{equation*}
where $c \in \amsbb{R}$ is a constant and $( L_{j} )_{j = 1}^{t}$ is a collection of bipartite operators constructed from the measurement operators of Alice and Bob. This immediately implies that for any state $\rho_{AB}$ we have $\tr ( W \rho_{AB} ) \leq c$. Moreover, if the SOS decomposition is tight, i.e.~$\beta_{Q} = c$, it yields explicit restrictions on the optimal realisation: any quantum realisation that achieves $\tr (W \rho_{AB}) = \beta_{Q}$ must satisfy
\begin{equation}
\label{eq:Lj-trace}
\tr ( L_{j}\hc L_{j} \rho_{AB}) = 0
\end{equation}
for all $j \in \{1, 2, \ldots, t\}$. Note that $\tr ( L_{j}\hc L_{j} \rho_{AB}) = \norm[\big]{ L_{j} \rho_{AB}^{1/2} }_{F}^{2}$, where $\norm{\cdot}_{F}$ is the Frobenius norm. Therefore, Eq.~\eqref{eq:Lj-trace} is equivalent to $L_{j} \rho_{AB}^{1/2} = 0$ and immediately implies $L_{j} \rho_{AB} = 0$.
\subsection{Two approaches to self-testing}
\label{sec:two-approaches}
Self-testing constitutes the most complete form of device-independent certification in which the quantum realisation is determined up to local unitaries and extra degrees of freedom. These two ambiguities cannot be resolved in a device-independent scenario and it is implicitly understood that in the task of self-testing we identify the quantum realisation up to these two equivalences. For brevity we will refer to them as the \emph{standard equivalences}. While for many scenarios these equivalences are sufficient, in some cases there is an extra equivalence resulting from the fact that the quantum realisation can be transposed to obtain an inequivalent realisation~\cite{mckague11a, kaniewski17a, andersson17a}. Since the transpose is as well-understood and unavoidable as the other two equivalences, we believe it is justified to also refer to this phenomenon as self-testing.

If we want to base a self-testing statement on the observed Bell violation, there are two approaches we can choose from. The starting point of a self-testing argument are algebraic relations satisfied by the optimal observables. Typically, in order to derive such relations from the observed Bell value, we examine the SOS decomposition (although in some cases one can also look at the square of the Bell operator~\cite{kaniewski17a}). In the first approach we use the knowledge of the ideal observables to propose a swap unitary. This unitary acts jointly on one part of the unknown state and a fresh register and attempts to swap out the relevant part of the state into the new register. Such a unitary is applied on both sides and we then use the algebraic relations to show that at the end the new registers hold the desired ideal state.

In this work, however, we use a different method proposed originally by Popescu and Rohrlich~\cite{popescu92a}. We use the previously deduced algebraic relations to derive relations which contain observables of a single party only. Since properties of observables can only be determined on the support of the reduced states $\rho_{A}$ and $\rho_{B}$, it is convenient to assume that the reduced states are full-rank (this assumption does not affect the conclusions, but it significantly simplifies the mathematical description of the problem). These single-party algebraic relations allow us to deduce the exact form of the local observables for each party up to local unitaries and extra degrees of freedom. Once the local observables have been characterised, we can simply construct the Bell operator and diagonalise it. The state shared by the players is now determined by the eigenspace corresponding to the largest eigenvalue.
\section{The modified Buhrman-Massar functional}
\label{sec:modified-BM}

We are now ready to present our main result. 
We begin with a new generalisation of the CHSH Bell inequality
to $d$-outcome Bell scenarios. 

In the CHSH scenario the settings $j, k$ and the outcomes $a, b$ are bits and the Bell functional is given by
\begin{equation*}
c_{abjk} :=
\begin{cases}
1/4 &\nbox{if} a \oplus b \oplus jk = 0,\\
0 &\nbox{otherwise.}
\end{cases}
\end{equation*}
It is well known that
%$\beta_{L} = \frac{3}{4}$, $\beta_{Q} = \frac{1}{2} + \frac{1}{2 \sqrt{2} }$
$\beta_{L} = 3/4$, $\beta_{Q} = 1/2 + 1/(2 \sqrt{2})$
and $\beta_{NS} = 1$. Buhrman and Massar proposed a natural generalisation of the CHSH functional by extending the input and output alphabet to $\{0, 1, \ldots, d - 1\}$ and replacing the XOR operation by addition modulo $d$~\cite{buhrman05a}. The BM functional is defined as
\begin{equation*}
c_{abjk} :=
\begin{cases}
1/d^{2} &\nbox{if} a + b + jk \equiv 0 \smod d,\\
0 &\nbox{otherwise.}
\end{cases}
\end{equation*}
The no-signalling value of this functional equals $\beta_{NS} = 1$ for all $d$~\cite{buhrman05a} and is achieved by a straightforward generalisation of the Popescu-Rohrlich box~\cite{popescu94a}. The quantum value is in general not known and the only analytic bound states that whenever $d$ is prime we have~\cite{bavarian15a}
\begin{equation}
\label{eq:analytic-upper-bound}
\beta_{Q} \leq \frac{1}{d} + \frac{ d - 1 }{ d \sqrt{d} }.
\end{equation}
For small values of $d$ upper and lower bounds have been computed numerically by Liang et al.~\cite{liang09a}. The only case in which the two bounds coincide corresponds to $d = 3$ and the resulting value is in excellent agreement with the analytic expression
\begin{equation*}
\beta_{Q} = \frac{1}{3} + \frac{ 2 \cos( \pi / 18 ) }{ 3 \sqrt{3} }
\end{equation*}
given in Ref.~\cite{ji08a}. The local value has been explicitly computed for prime $d$ up to $d = 13$~\cite{liang09a}, but no analytic formula is known.
%
%While this is a valid functional for any integer $d$, whenever $d$ is a prime the problem acquires additional structures as we can interpret the settings and the outcomes as elements of a finite field. As a consequence most follow-up works only consider the case of $d$ being a prime power.
%The BM operator written in terms of observables reads
%
%\begin{equation*}
%
%W := \frac{1}{d^{3}} \sum_{n = 0}^{d - 1} \sum_{jk} \omega^{njk} A_{j}^{(n)} \otimes B_{k}^{(n)},
%
%\end{equation*}

Although the BM functional is clearly a natural generalisation of the CHSH functional, its quantum value seems hard to determine. To avoid this problem, we propose a modification of the BM functional for which the quantum value can be computed analytically. Writing the BM operator in terms of operators $A_{j}^{(n)}$ and $B_{k}^{(n)}$ yields
\begin{equation*}
\frac{1}{d^{3}} \sum_{n = 0}^{d - 1} \sum_{jk} \omega^{njk} A_{j}^{(n)} \otimes B_{k}^{(n)}.
\end{equation*}
%
%and it is easy to verify that the Fourier coefficients satisfy condition~\eqref{eq:u-condition}.
We consider a generalisation of this Bell operator given by
\begin{equation}
\label{eq:Wd-definition}
W_{d} := \frac{1}{d^{3}} \sum_{n = 0}^{d - 1} \lambda_{n} \sum_{jk} \omega^{njk} A_{j}^{(n)} \otimes B_{k}^{(n)},
\end{equation}
where $\lambda_{n}$ are complex numbers of unit modulus. To ensure that condition~\eqref{eq:u-condition} is satisfied, we must choose $\lambda_{0}$ to be real, i.e.~$\lambda_{0} = 1$, and $\lambda_{n} = \lambda_{d - n}^{*}$ for $n \in \{1, 2, \ldots, d - 1\}$.

In the remainder of this section we show that for every prime dimension $d \geq 3$ there exists a valid choice of phases $\lambda_{n}$ for which the quantum value can be computed analytically. In the first step we show that regardless of the choice of $\lambda_{n}$ we have
\begin{equation}
\label{eq:Wd-upper-bound}
W_{d} \leq \bigg( \frac{1}{d} + \frac{ d - 1 }{d \sqrt{d}} \bigg) \, \mathbb{1},
\end{equation}
which means that the upper bound given in Eq.~\eqref{eq:analytic-upper-bound} holds for arbitrary phase factors. In the second step we specify the phases and give a quantum realisation which saturates this bound.

To prove the generic upper bound given in Eq.~\eqref{eq:Wd-upper-bound} we construct a SOS decomposition. For an integer $n$ satisfying $\abs{n} \leq d - 1$ define
\begin{equation}
\label{eq:Cjn-definition}
\Cjndef,
\end{equation}
where $\lambda_{-n} := \lambda_{d - n}$. It is easy to check that $\big[ C_{j}^{(n)} \big]\hc = C_{j}^{(-n)}$, which allows us to write the Bell operator as
\begin{equation*}
W_{d} = \frac{ 1 }{d^{2} \sqrt{d}} \sum_{n = 0}^{d - 1} \sum_{j} A_{j}^{(n)} \otimes C_{j}^{(n)}.
\end{equation*}
The term corresponding to $n = 0$ is proportional to identity, whereas terms $n$ and $d - n$ are conjugate to each other. Therefore, it is convenient to write the Bell operator as
\begin{equation}
\label{eq:Wd-sum-Tn}
W_{d} = \frac{ \mathbb{1} }{d} + \frac{ 1 }{d^{2} \sqrt{d}} \sum_{n = 1}^{(d - 1)/2} T_{n}
\end{equation}
for
\begin{equation*}
T_{n} := \sum_{j} A_{j}^{(n)} \otimes C_{j}^{(n)} + \sum_{j} A_{j}^{(-n)} \otimes C_{j}^{(-n)}.
\end{equation*}
(Note that since $d$ is odd, $(d-1)/2$ is an integer.)
This expression can be rewritten as 
\begin{equation}
\begin{aligned}
\label{eq:Tn-SOS}
T_{n} = &\sum_{j} A_{j}^{(n)} A_{j}^{(-n)} \otimes \mathbb{1} + \mathbb{1} \otimes \sum_{k} B_{k}^{(-n)} B_{k}^{(n)}\\
&- \sum_{j} \big[ L_{j}^{(n)} \big] \hc L_{j}^{(n)}
\end{aligned}
\end{equation}
for $L_{j}^{(n)} := A_{j}^{(-n)} \otimes \mathbb{1} - \mathbb{1} \otimes C_{j}^{(n)}$. Since $A_{j}^{(n)} A_{j}^{(-n)} \leq \mathbb{1}$ and $B_{k}^{(-n)} B_{k}^{(n)} \leq \mathbb{1}$, we conclude that
\begin{equation}
\label{eq:Tn-upper-bound}
T_{n} \leq  2 d \, \mathbb{1},
\end{equation}
which gives the desired operator bound.

Before proceeding to the second step, let us make two comments about this SOS decomposition. First note that so far we have only used the fact that $d$ is odd. However, the decomposition easily generalises to the case of even $d$ (the only difference being that the operator $T_{n}$ corresponding to $n = d/2$ is a single sum of Hermitian operators). This shows that the upper bound given in Eq.~\eqref{eq:analytic-upper-bound} holds for all dimensions $d$ (not necessarily prime). Moreover, note that the SOS decomposition can be straightforwardly generalised to the case where the coefficients $\lambda_{n}$ are arbitrary complex numbers satisfying $\lambda_{n}=\lambda_{d-n}^{*}$. 

%While the resulting upper bound is not tight for the trivial choice of $\lambda_{n} = 1$, it becomes tight if we choose these phases correctly.
%It is then easy to see that this decomposition applies well to a modification of the BM functional in which 
%this upper bound holds for all integers $d$. \jed{what is the connection to the first level of NPA? can we argue from the SOS that this bound is not achievable?}
%
To show that the upper bound given in Eq.~\eqref{eq:Wd-upper-bound} can be saturated, we specify the phases and give an explicit quantum realisation. For an odd prime $d$ the phases are chosen as (see Appendix \ref{app:fazy} for details)
\begin{equation}\label{eq:phases1}
\lambda_{n} := \Big[ \varepsilon_{d} \Big( \frac{n}{d} \Big) \Big]^{-1} \omega^{- g(n, d)/48}
\end{equation}
where
\begin{equation}\label{eq:epsilon}
\varepsilon_{d} :=
\begin{cases}
1 &\nbox{if} d \equiv 1 \smod 4,\\
i &\nbox{if} d \equiv 3 \smod 4,
\end{cases}
\end{equation}
$\big( \frac{n}{d} \big)$ is the Legendre symbol\footnote{Recall that the Legendre symbol $\big( \frac{n}{d} \big)$ equals $+1$ if $n$ is a quadratic residue modulo $d$ and $-1$ otherwise.} and
\begin{widetext}
\begin{equation}\label{eq:phases12}
g(n, d) :=
\begin{cases}
n \big[ n^{2} - d ( d + 6 ) + 3 \big] &\nbox{if} n \equiv 0 \smod 2 \nbox{and} (n + d + 1)/2 \equiv 0 \smod 2,\\
n \big[ n^{2} - d ( d - 6 ) + 3 \big] &\nbox{if} n \equiv 0 \smod 2 \nbox{and} (n + d + 1)/2 \equiv 1 \smod 2,\\
n ( n^{2} + 3 ) + 2 d^{2} ( -5n + 3 ) &\nbox{if} n \equiv 1 \smod 4,\\
n ( n^{2} + 3 ) + 2 d^{2} ( n + 3 ) &\nbox{if} n \equiv 3 \smod 4.
\end{cases}
\end{equation}
\end{widetext}
%
%Let us point out that this SOS decomposition is valid for arbitrary values of $\lambda_{n}$ and, therefore, we should not expect it to be tight for all choices of phases. In fact, it is quite remarkable that whenever $d$ is an odd prime there exists a choice of $\lambda_{n}$ which makes this upper bound tight. The choice of phases is based on the
The optimal quantum realisation is inspired by that of Ji et al.~\cite{ji08a}: Alice and Bob share a maximally entangled state of local dimension $d$, 
\begin{equation}\label{eq:maxent}
\ket{\Phi}_{AB} := \frac{1}{\sqrt{d}} \sum_{j = 0}^{d - 1} \ket{j}_{A} \ket{j}_{B},
\end{equation}
and perform rank-1 projective measurements which are (pairwise) mutually unbiased. Since the measurements are projective, they are fully determined by a single unitary operator: the observable. These are conveniently expressed in terms of the Heisenberg-\!Weyl operators:
\begin{equation*}
X := \sum_{j} \ketbra{ j + 1 }{j} \nbox{and} Z := \sum_{j} \omega^{j} \ketbraq{j},
\end{equation*}
where the summation goes over $j \in \{0, 1, \ldots, d - 1\}$ and $\ket{d} \equiv \ket{0}$. The observable corresponding to the $k$-th measurement of Bob is given by
\begin{equation}
\label{eq:Bk-definition}
%
%A_{j} &:= \frac{ \lambda_{1}^{*} }{\sqrt{d}} \sum_{k} \omega^{-jk - k (k + 1) } X Z^{-k},\\
%
\bobobservables.
\end{equation}
It is straightforward to check that these are valid observables (they are clearly unitary, while the correctness of the spectrum can be verified by showing that $B_{k}^{d} = \mathbb{1}$), and therefore $B_k^{(n)}=B_k^{n}$. To see that these observables correspond to mutually unbiased bases note that if we disregard the multiplicative factor $\omega^{k(k+1)}$, which corresponds to a cyclic shift of the outcomes, we obtain one of the standard constructions of a complete set of mutually unbiased bases in prime dimension~\cite{bandyopadhyay02a}.

Before defining the observables of Alice, it is convenient to explore some properties of the observables of Bob. If we compute the corresponding operator $C_{j}^{(1)}$ (as defined in Eq.~\eqref{eq:Cjn-definition}), we find that $C_{j}^{(1)}$ itself is a valid observable (i.e.~it is unitary and has the correct spectrum). Moreover, the higher-order operators satisfy
\begin{equation}
\label{eq:power-relation}
\powerrelation
\end{equation}
for $n \in \{2, 3, \ldots, d - 1\}$ (see Appendix~\ref{app:fazy} for details). The first observation enables us to define the observables of Alice as
\begin{equation*}
A_{j} := \big[ C_{j}^{(1)} \big]^{*} = \frac{ \lambda_{1}^{*} }{\sqrt{d}} \sum_{k} \omega^{-jk - k (k + 1) } X Z^{-k},
\end{equation*}
where $^{*}$ denotes the complex conjugation in the standard basis (which does not affect the spectrum of the operator). The power relation given in Eq.~\eqref{eq:power-relation} allows us to evaluate all the terms appearing in the Bell functional:
\begin{align*}
\bramatketq{ \Phi_{AB} }{ &A_{j}^{(n)} \otimes C_{j}^{(n)} } = \bramatketq{ \Phi_{AB} }{ A_{j}^{n} \otimes \big[ C_{j}^{(1)} \big]^{n} }\\
&= \bramatketq{ \Phi_{AB} }{ A_{j}^{n} \otimes \big[ A_{j}^{*} \big]^{n} } = \frac{1}{d} \tr \big( A_{j}^{n} \big[ A_{j}\hc \big]^{n} \big) = 1.
%
%\bramatketq{\Phi}{ A_{j}^{(n)} \otimes C_{j}^{(n)} } = \bramatketq{\Phi}{ A_{j}^{n} \otimes \big[ C_{j}^{(1)} \big]^{n} }
%
%\frac{1}{d} \tr \big( A_{j}\tran C_{j}^{(1)} \big) = \frac{1}{d} \tr \Big( \big[ C_{j}^{(1)} \big]\hc C_{j}^{(1)} \Big) = 1.
%
\end{align*}
This immediately implies that
%
%This implies that $\bramatketq{ \Phi }{ T_{1} } = 2$, which saturates the upper bound given in Eq.~\eqref{eq:Tn-upper-bound}. Analogous results for $n \geq 2$ follows from the fact that for projective observables we have $A_{j}^{(n)} = A_{j}^{n}$ combined with Eq.~\eqref{eq:power-relation}. This implies that this quantum realisation achieves
%
\begin{equation*}
\bramatketq{ \Phi_{AB} }{ W_{d} } = \frac{1}{d} + \frac{ d - 1 }{d \sqrt{d}},
\end{equation*}
which saturates the upper bound given in Eq.~\eqref{eq:Wd-upper-bound}.

Since this quantum realisation saturates the Bell expression term-by-term, it is clear that the same will be true if we rescale terms corresponding to distinct $n$ by arbitrary non-negative numbers (the scaling must preserve the condition~\eqref{eq:u-condition}). In other words, we have obtained a whole family of Bell inequalities saturated by precisely the same set of probability points. For certification purposes in the exact case, i.e.~when the violation is maximal, all these inequalities are equivalent. In the presence of noise they will not be the same, but it is not clear what the optimal choice of weights is.
%In the definition of the Bell operator given in Eq.~\eqref{eq:Wd-definition} we have restricted the coefficients $\lambda_{n}$ to be of unit modulus, but since in the SOS decomposition we have distinct terms corresponding to different values of $n$ and similarly the quantum realisation saturates the each of these terms separately, the results immediately generalise to the case where $\lambda_{n}$ have .

Having presented the quantum realisation saturating the upper bound, it is easy to see how the phases $\lambda_{n}$ were computed. We started from the observables of Bob given in Eq.~\eqref{eq:Bk-definition} and looked for phases which would give the resulting operators $C_{j}^{(n)}$ the desired properties (unitarity and correct spectrum of $C_{j}^{(1)}$ and the power relation specified in Eq.~\eqref{eq:power-relation}). Clearly, for a generic choice of observables of Bob this would not be possible. This shows that the measurements proposed by Ji et al.~\cite{ji08a}, which combine the standard MUB construction with carefully chosen prefactors, constitute a rather special choice.
%
%As a consequence for every odd $d \geq 3$ we have a Bell functional for which the quantum value is known and we have an explicit quantum realisation achieving it. It is worth pointing that for $d = 3$ the unique probability point that achieves the quantum value of the modified-BM functional also saturates the BM functional. However, the numerical results from Ref.~\cite{liang09a} suggest that this property does not hold in general. More specifically, for $d = 5, 11, 13$ the quantum value of the BM functional appears to be strictly larger than what can be achieved using the maximally entangled state of dimension $d$.
%
%\jed{check}One might wonder whether the upper bound~\eqref{eq:Wd-upper-bound} remains tight when $d$ is not prime. Let us conclude with a brief discussion of the case of non-prime dimensions. The optimal quantum realisations for prime $d$ use measurements which pairwise mutually unbiased. Therefore, we might intuitively expect that we must find $d$ MUBs in dimension $d$ which is hard in dimensions which are not prime powers. Extending the results to the case of prime powers seems more feasible, but we leave it for future work.

%\subsection{Classical values for small $d$}

We have shown how to achieve the quantum value using quantum systems of local dimension $d$, but it is not a priori clear that one cannot achieve it using quantum systems of lower dimension. To gain some intuition we have performed numerical search over quantum strategies of local dimension $r < d$ for $d = 3, 5$. We have used the standard see-saw procedure in which we pick a random quantum realisation (states and measurements) on $\amsbb{C}^{r} \otimes \amsbb{C}^{r}$ and optimise until we reach a local maximum. This procedure is not guaranteed to converge to the global maximum, but given a sufficient number of repetitions one can hope to explore the entire landscape of realisations in a fixed dimension. Let $\beta_{Q}^{r}$ be the largest value achievable using quantum systems of local dimension $r$. Clearly $\beta_{Q}^{1} = \beta_{L}$, whereas for our inequalities $\beta_{Q}^{d} = \beta_{Q}$, therefore, we only consider $r \in \{ 2, 3, \ldots, d - 1 \}$. For completeness for $d = 3, 5, 7$ we have also computed the classical value (by enumerating all deterministic strategies).

%The modified functional can be maximally violated with the maximally entangled state of local dimension $d$ and, therefore, one could hope that it can be used as a dimension witness. To provide a partial answer to this question for functionals corresponding to $d = 3, 5, 7$ we have performed numerical search over states of bounded dimension (more specifically, we have restricted the Schmidt-rank of the state).
%
%In Section~\ref{sec:} we have computed the quantum value of the functional correspoding to $d$ outcomes (for $d$ being prime) and we have shown how the maximal violation can be achieved using the maximally entangled state of dimension $d$. We conjecture that this Bell functional is a self-test, i.e.~there is essentialy just one quantum realisations (up to well-understood equivalences) that achieves it. In particular, we conjecture that one requires systems of local dimension $d$. To support this conjecture we have performed a numerical search over strategies with fixed local dimension (more specifically, we have enforced the state to have a bounded Schmidt rank).

For $d = 3$ we have
\begin{equation*}
\beta_{L} = \frac{1}{3} + \frac{ 2 \cos (\pi/9) }{ 3 \sqrt{3} } \approx 0.6950, \quad \beta_{Q} = \frac{1}{3} + \frac{ 2 }{ 3 \sqrt{3} } \approx 0.7182.
\end{equation*}
The numerical search over two-qubit strategies did not yield a single instance exceeding the classical value. Therefore, we conjecture that one cannot violate the $d = 3$ inequality using qubits (similar to the inequality proposed in Ref. \cite{lim10a}).

For $d = 5$ we have
\begin{equation*}
\beta_{L} = \frac{9}{25} + \frac{8}{ 25 \sqrt{5} } \approx 0.5031, \quad \beta_{Q} = \frac{1}{5} + \frac{4}{ 5 \sqrt{5} } \approx 0.5578.
\end{equation*}
The numerical search suggests that
\begin{equation*}
\beta_{Q}^{2} = 0.5100, \quad \beta_{Q}^{3} = \beta_{Q}^{4} = 0.5373.
\end{equation*}
%
% for the record, the original values given by Jordi in the other formulation are:
% \beta_{L} = 8 + 4 \sqrt{5} \approx 16.9443
% \beta_{Q}^{2} \approx 17.3292
% \beta_{Q}^{4} \approx 18.857
%
Interestingly enough, setting $r = 4$ always leads to a solution where the state is of Schmidt-rank 3.
%where two Schmidt coefficients are equal. Restricting the Schmidt-rank to be at most 2 yields a state whose Schmidt coefficients are distinct.

%, we obtain ... \jed{check if true} If we do not restrict the dimension, then every time the optimisation converges to the optimal value $\beta_{Q}$ the realisation is equivalent to the one specified before, which supports the conjecture that the inequality is a self-test.

For $d = 7$ we have\footnote{The analytic expression for $\beta_L$ is quite complicated, so we only give an approximate value.}
\begin{equation*}
\beta_{L} \approx 0.4001, \quad \beta_{Q} = \frac{1}{7} + \frac{6}{ 7 \sqrt{7} } \approx 0.4668.
\end{equation*}
One might ask whether these Bell inequalities correspond to facets of the local set, but we believe that this is not the case. For $d = 3$ and $d = 5$ we have found all the deterministic points saturating the local bound (9 and 125 points, respectively) and these did not correspond to a facet. We conjecture that the inequalities for larger values of $d$ behave in analogous manner.
\section{Quantum realisations which achieve the maximal violation}
\label{sec:quantum-realisations}
In the previous section we have proposed a new Bell functional and shown that the quantum value can be achieved by performing mutually unbiased measurements on the maximally entangled state of dimension $d$. A natural follow-up question is whether this Bell functional is a self-test, i.e.~whether this is the only manner of achieving the quantum value (up to the standard equivalences).

In this section we show that the optimal observables are fully characterised by simple algebraic relations. For $d = 3$ these conditions are sufficient to explicitly derive the form of the observables, which leads to a complete self-testing statement (although the extra freedom coming from the transposition map must be included). On the other hand, for $d = 5$ and $d = 7$ we have found additional, inequivalent quantum realisations, which implies that the introduced Bell functionals are not self-tests.% in the usual sense.
\subsection{Necessary and sufficient algebraic conditions}
\label{sec:neccessary-and-sufficient}
As explained in Section~\ref{sec:bell-operator} a tight SOS decomposition leads to explicit algebraic conditions that every optimal realisation must satisfy. To achieve the quantum value every term in Eq.~\eqref{eq:Wd-sum-Tn} must be saturated, i.e.~$\tr (T_{n} \rho_{AB}) = 2 d$ for all $n \in \{ 1, 2, \ldots, (d	 - 1)/2 \}$. By the SOS decomposition given in Eq.~\eqref{eq:Tn-SOS} this implies that
\begin{gather}
\label{eq:Aj-condition}
\tr \left[ A_{j}^{(n)} A_{j}^{(-n)} \rho_{A} \right] = 1,\\
\label{eq:Bk-condition}
\tr \left[ B_{k}^{(-n)} B_{k}^{(n)} \rho_{B} \right] = 1,\\
\label{eq:Lj-condition}
L_{j}^{(n)} \rho_{AB} = 0.
\end{gather}
for all $j, k \in \{0, 1, \ldots, d - 1\}$ and $n \in \{ 1, 2, \ldots, d - 1 \}$. It is easy to see that swapping the roles of $L_{j}^{(n)}$ and $\big[ L_{j}^{(n)}\big] \hc$ leads to an analogous SOS decomposition, which implies
\begin{equation}
\label{eq:Ljhc-condition}
\big[ L_{j}^{(n)} \big]\hc \rho_{AB} = 0.
\end{equation}
Conditions~\eqref{eq:Aj-condition} and~\eqref{eq:Bk-condition} simply require that the observables of Alice and Bob correspond to measurements which are projective on the (local) support of the state. Under the assumption that the reduced states are full-rank, we deduce that all the measurements are projective. As explained in Section~\ref{sec:measurements-observables} this allows us to write $A_{j}^{(n)} = A_{j}^{n}$ for $A_{j} := A_{j}^{(1)}$, where $A_{j}$ is a unitary satisfying $A_{j}^{d} = \mathbb{1}$. Similarly $B_{k}^{(n)} = B_{k}^{n}$ for $B_{k} := B_{k}^{(1)}$, which is unitary and satisfies $B_{k}^{d} = \mathbb{1}$.
%(and recall that $A_{j}$ and $B_{k}$ are unitaries).
%$A_{j}\hc A_{j} = A_{j} A_{j}\hc = \mathbb{1}$ and $B_{k}\hc B_{k} = B_{k} B_{k}\hc = \mathbb{1}$).

Conditions~\eqref{eq:Lj-condition} and~\eqref{eq:Ljhc-condition}, on the other hand, imply some relations between the observables of Alice and Bob, namely:
\begin{gather}
\label{eq:Aj-Cj-relation-1}
( A_{j}^{-n} \otimes \mathbb{1} ) \rho_{AB} = ( \mathbb{1} \otimes C_{j}^{(n)} ) \rho_{AB},\\
\label{eq:Aj-Cj-relation-2}
( A_{j}^{n} \otimes \mathbb{1} ) \rho_{AB} = ( \mathbb{1} \otimes C_{j}^{(-n)} ) \rho_{AB}.
\end{gather}
for $n \in \{1, 2, \ldots, d - 1\}$. Our goal now is to infer what impact these two relations have on the form of the operators $C_{j}^{(n)}$. Let us start with $n = 1$. The fact that the observables of Alice are unitary implies that
\begin{align*}
\rho_{AB} &= ( A_{j} A_{j}^{-1} \otimes \mathbb{1} ) \rho_{AB} = ( A_{j} \otimes C_{j}^{(1)} ) \rho_{AB}\\
&= ( \mathbb{1} \otimes C_{j}^{(1)} C_{j}^{(-1)} ) \rho_{AB}.
\end{align*}
Tracing out the subsystem of Alice and using the fact that $\rho_{B}$ is full-rank gives $C_{j}^{(1)} C_{j}^{(-1)} = \mathbb{1}$. Since $C_{j}^{(-1)} = \big[ C_{j}^{(1)} \big]\hc$, this implies that for all $j \in \{ 0, 1, \ldots, d - 1 \}$ the operators $C_{j}^{(1)}$ are unitary. Moreover, since $A_{j}^{-d} = \mathbb{1}$, we have
\begin{align*}
\rho_{AB} &= ( A_{j}^{-d} \otimes \mathbb{1} ) \rho_{AB} = \big( \mathbb{1} \otimes \big[ C_{j}^{(1)} \big]^{d} \big) \rho_{AB}
\end{align*}
and we conclude that $\big[ C_{j}^{(1)} \big]^{d} = \mathbb{1}$, i.e.~it has the correct spectrum.
%
%It turns out that these relations can be combined to completely eliminate the observables of one party.
%i.e.~$\tr ( A_{j}\hc A_{j} \rho_{A} ) = 1$ for all $j$ and $\tr ( B_{k} B_{k}\hc \rho_{B} ) = 1$ for all $k$. Under the assumption that $\rho_{A}$ and $\rho_{B}$ are full-rank this actually implies that $A_{j}\hc A_{j} = \mathbb{1}$ and $B_{k} B_{k}\hc = \mathbb{1}$, i.e.~that all the observables are projective. Writing out the last condition gives
%
%An analogous SOS can be obtained by swapping the roles of $L_{j}$ and $L_{j}\hc$, which leads to a family of complementary conditions $L_{j}\hc \rho_{AB} = 0$, i.e.
%
%\begin{equation*}
%
%( A_{j}\hc \otimes \mathbb{1} ) \rho_{AB} = e^{ i \phi_{d} } ( \mathbb{1} \otimes C_{j} ) \rho_{AB}.
%
%\end{equation*}
%
%We have already established that the observables of Alice are projective, i.e.~$A_{j}^{(n)} = A_{j}^{n}$.
Moreover, for an arbitrary integer $t \in \{ 1, 2, \ldots, d - 1 \}$ we can write
\begin{equation*}
( \mathbb{1} \otimes C_{j}^{(t)} ) \rho_{AB} = ( A_{j}^{-t} \otimes \mathbb{1} ) \rho_{AB} = ( \mathbb{1} \otimes \big[ C_{j}^{(1)} \big]^{t} ) \rho_{AB},
\end{equation*}
where we have used the relation~\eqref{eq:Aj-Cj-relation-1} twice: for $n = t$ and $n = 1$. This immediately implies that
%
%\begin{equation*}
%
%C_{j}^{(t)} \rho_{B} = C_{j}^{t} \rho_{B},
%
%\end{equation*}
%
%which under the assumption
%
\begin{equation}
\label{eq:power-relation-2}
C_{j}^{(t)} = \big[ C_{j}^{(1)} \big]^{t}
\end{equation}
for all $t \in \{1, 2, \ldots, n - 1\}$.
%
%Tracing out the system of Alice gives
%
%\begin{equation*}
%
%C_{j}\hc \rho_{B} = e^{ i \phi_{d} d } C_{j}^{d - 1} \rho_{B},
%
%\end{equation*}
%
%but since $\rho_{B}$ is full-rank, this actually implies
%
%\begin{equation}
%\label{eq:algebraic-relation}
%
%C_{j}\hc = e^{ i \phi_{d} d } C_{j}^{d - 1}.
%
%\end{equation}
%
%This equality is equivalent to requiring that for all $j$ the operator $C_{j}$ is unitary and satisfies $e^{ i \phi_{d} d } C_{j}^{d} = \mathbb{1}$.

In the previous paragraph we have established that the maximal violation requires the operators $C_{j}^{(n)}$ to be unitary, have the correct spectrum and satisfy the power relation~\eqref{eq:power-relation-2}. Since the operators $C_{j}^{(n)}$ are constructed out of the observables of Bob, these conditions are restrictions on the observables $B_{k}$. However, it follows immediately from the calculations presented in Section~\ref{sec:modified-BM} that any observables satisfying these conditions can be used to construct the entire quantum realisation: we simply take the maximally entangled state and define the observables of Alice as $A_{j} := \big[ C_{j}^{(1)} \big]^{*}$. This shows that this characterisation is tight: a set of observables of Bob is capable of producing the maximal violation (on a state that is locally full rank) if and only if the resulting operators $C_{j}^{(n)}$ are unitary, have the correct spectrum and satisfy the power relation.
\subsection{A complete self-testing statement for $d = 3$}
\label{sec:self-testing}
For $d = 3$ the algebraic characterisation is sufficient to derive the explicit form of the observables. In this case the condition $\big[ C_{j}^{(1)} \big]^{3} = \mathbb{1}$ is equivalent to
\begin{equation*}
\big[ C_{j}^{(1)} \big]\hc = \big[ C_{j}^{(1)} \big]^{2},
\end{equation*}
which can be rewritten as
\begin{equation*}
\frac{ \lambda_{1}^{*} }{\sqrt{3}} \sum_{k} \omega^{-jk} B_{k}\hc = \frac{ \lambda_{1}^{2} }{3} \sum_{kk'} \omega^{j (k + k')} B_{k} B_{k'}
\end{equation*}
for $\om[]{3}$ and $\lambda_{1} = e^{ - \mathbbm{i} \pi / 18}$. Since the observables of Bob are projective, we have $B_{j}^{2} = B_{j}\hc$, which leads to
\begin{equation}
\label{eq:Bkdagger-anticommutators}
- \omega^{2} \sum_{k} \omega^{-jk} B_{k}\hc = \sum_{k \neq k'} \omega^{j (k + k')} B_{k} B_{k'},
\end{equation}
where
%the summation on the right-hand side goes over pairs $(k, k')$ such that $k \neq k'$ and
we have used the fact that
\begin{equation*}
\frac{ \lambda_{1}^{*} \sqrt{3} }{ \lambda_{1}^{2} } - 1 = - \omega^{2}.
\end{equation*}
By taking suitable linear combinations of Eq.~\eqref{eq:Bkdagger-anticommutators} corresponding to distinct values of $j \in \{0, 1, 2\}$ we arrive at
\begin{align*}
\comrelation{0}{1}{2},\\
\comrelation{2}{0}{1},\\
\comrelation{1}{2}{0}
\end{align*}
and these relations turn out to be sufficient to reconstruct the observables of Bob.

As explained in Section~\ref{sec:two-approaches} sometimes the standard equivalences must be supplemented by the freedom resulting from the transposition map and this is precisely what happens in this case. In Appendix~\ref{app:optimal-observables-d=3} we show that projectivity and the commutation relations above imply that the Hilbert space of Bob $\sH_{B}$ contains a qutrit, i.e.~$\sH_{B} \equiv \sH_{B'} \otimes \sH_{B''}$ for $\sH_{B'} \equiv \amsbb{C}^{3}$. Moreover, one can find a unitary $U_{B} : \sH_{B} \to \sH_{B'} \otimes \sH_{B''}$ such that
\begin{equation*}
\bobobservablesBM,
\end{equation*}
where the \emph{canonical observables} are given by
\canonicalobservables
and $Q_{1}, Q_{2}$ are orthogonal projectors satisfying $Q_{1} + Q_{2} = \mathbb{1}_{B''}$. These two projectors identify the orthogonal subspaces corresponding to the two inequivalent solutions. Since the Bell functional is symmetric, we obtain analogous relations for the observables of Alice: we conclude that $\sH_{A} \equiv \sH_{A'} \otimes \sH_{A''}$ for $\sH_{A'} \equiv \amsbb{C}^{3}$ and that one can find a unitary $U_{A} : \sH_{A} \to \sH_{A'} \otimes \sH_{A''}$ such that
\begin{equation*}
\aliceobservables,
\end{equation*}
where $P_{1}, P_{2}$ are orthogonal projectors satisfying $P_{1} + P_{2} = \mathbb{1}_{A''}$. This characterisation allows us to write down the Bell operator as
\begin{equation}
U W_{3} U\hc = \sum_{xy} W_{xy} \otimes P_{x} \otimes Q_{y},
\end{equation}
where $U := U_{A} \otimes U_{B}$, the summation goes over $x, y \in \{0, 1\}$ and $W_{xy}$ is the two-qutrit Bell operator corresponding to the canonical observables $A_{j} = O_{j}^{(x)}$ and $B_{k} = O_{k}^{(y)}$. The two-qutrit Bell operators can be diagonalised explicitly and we find that only $W_{01}$ and $W_{10}$ contain $\mu = \frac{1}{3} + \frac{ 2 }{ 3 \sqrt{3} }$ as an eigenvalue. In both cases the corresponding eigenspace is 1-dimensional and the eigenvector is simply the maximally entangled state of two qutrits: $\ket{\Phi} := ( \ket{00} + \ket{11} + \ket{22} ) /\sqrt{3}$. It follows that any state that achieves the maximal violation must be of the form
\begin{equation}
\label{eq:optimal-state}
\state,
\end{equation}
where $\sigma_{A''B''}$ is an arbitrary state satisfying
\begin{equation}
\label{eq:trace-projectors}
\tr \big[ ( P_{0} \otimes Q_{1} + P_{1} \otimes Q_{0} ) \sigma_{A''B''} \big] = 1.
\end{equation}
Condition~\eqref{eq:optimal-state} implies that the maximal violation certifies the maximally entangled state of two qutrits, which can be extracted by tracing out the auxiliary registers $A''$ and $B''$. Condition~\eqref{eq:trace-projectors} shows that the maximal violation is only possible when the observables of Alice and Bob belong to the two inequivalent classes.

This self-testing result has a couple of immediate consequences. First of all, the maximal violation is achieved by a single probability point in the quantum set of correlations, which implies that this is an exposed\footnote{A point in the quantum set of correlations is called \textit{exposed} if it is the unique maximiser of some Bell functional. Note, however, that although every exposed point is extremal, the converse does not hold (see Ref.~\cite{goh18a} for an example of an extremal but not exposed point of the quantum set).} point of the quantum set. Moreover, the marginal distributions of outcomes are uniform and it is easy to see that they remain uniform even for an eavesdropper who holds a purification. Intuitively speaking, this comes from the fact that the randomness is produced from the pure entangled state $\Phi_{A'B'}$, whereas the adversary can only hold the purification of $\sigma_{A''B''}$. The maximal violation certifies $\log 3$ bits of local randomness (for each input of Alice or Bob) against an external adversary, which makes this Bell inequality a good candidate for cryptographic tasks like generation of certified randomness or secret key.
\subsection{Inequivalent quantum realisations for $d = 5$ and $d = 7$}
Having fully solved the $d = 3$ case one might expect to obtain similar results for higher dimensions. Perhaps surprisingly, this turns out not to be the case: we have found that in dimensions $d = 5$ and $d = 7$ there exist additional, inequivalent choices of local observables which give rise to the maximal violation. The construction is a simple generalisation of the original observables given in Eq.~\eqref{eq:Bk-definition}. It is easy to check that for arbitrary $q \in \{1, 2, \ldots, d - 1\}$ and arbitrary function $h : \{0, 1, \ldots, d - 1\} \to \{0, 1, \ldots, d - 1\}$ the operators 
\begin{equation}
\label{eq:Bk-generalised}
B_{k} := \omega^{ h(k) } X Z^{qk}.
\end{equation}
constitute a valid set of observables.
We have found that for $d = 5, 7$ for every $q = \{1, 2, \ldots, d - 1\}$ there exists a function $h$ which ensures that these observables of Bob give rise to valid operators $C_{j}^{(n)}$. To see that these are not equivalent to the original observables, it suffices to look at commutation relations: the observables defined above satisfy
\begin{equation*}
\omega^{ q } B_{0} B_{1} = B_{1} B_{0}.
\end{equation*}
On the other hand, the original observables from Eq.~\eqref{eq:Bk-definition} satisfy
\begin{equation*}
\omega B_{0} B_{1} = B_{1} B_{0},
\end{equation*}
whereas their transposes satisfy
\begin{equation*}
\omega^{d-1} B_{0}\tran B_{1}\tran = B_{1}\tran B_{0}\tran.
\end{equation*}
Clearly, whenever $d \geq 5$ choosing $q = 2$ in Eq.~\eqref{eq:Bk-generalised} gives rise to a solution which is neither unitarily equivalent to the original realisation nor to its transpose. Nevertheless, all these realisations use the maximally entangled state of local dimension $d$. Therefore, it is possible that the maximal violation certifies the state, but not the measurements. Finally, let us point out that these distinct quantum realisations lead to the same probability point. Therefore, one might conjecture that despite the ambiguity at the level of quantum realisations, the Bell functional is maximised by a single probability point.
%
%In fact, there is a much more direct way of seeing that these correspond to inequivalent quantum realisations. Given a particular choice of optimal observables of Bob, we can reconstruct the entire quantum realisation and compute the probabilities of observing different outcomes. It turns out that for $d = 5$ and $d = 7$ the different choices of $q$ lead to distinct probability points. For $d = 5$ we obtain two distinct points, i.e.~the corresponding face of the quantum set is at least a line. For $d = 7$ we obtain three distinct points, i.e.~the corresponding face is at least a two-dimensional plane. This shows that the case $d \geq 5$ is qualitatively different from $d = 3$ and, therefore, we leave its analysis for future work.
%
\section{Conclusions}
\label{sec:conclusions}
The Buhrman-Massar generalisation of the CHSH inequality, despite its apparent simplicity, turns out to be hard to analyse. In particular, despite both analytical and numerical studies the behaviour of its quantum value is not known. In this work we propose a simple modification which allows us to analyse the resulting functional. More specifically, we propose a family of Bell functionals labelled by prime $d \geq 3$, whose quantum value can be determined analytically. For every such $d$ we give an explicit realisation which achieves the quantum value in which Alice and Bob share the maximally entangled state of local dimension $d$ and perform local rank-1 projective measurements which are pairwise mutually unbiased. We thus generalise the CHSH Bell inequality to $d$-outcome Bell scenarios, preserving at the same time
the most relevant  of its properties: (i) analytical computability of its maximal quantum value and (ii) achievability of the maximal quantum violation by the maximally entangled state and mutually unbiased bases.

Once we know the quantum value and we have a particular quantum realisation of it, one might ask whether this realisation is unique (up to some well-understood equivalences). The SOS decomposition yields explicit algebraic relations that the local measurements must satisfy. For $d = 3$ these can be fully resolved: the quantum realisation is unique up to extra degrees of freedom, local unitaries and transposition. Unfortunately, the situation becomes more complicated for higher $d$. For $d = 5$ and $d = 7$ we have found alternative realisations which despite apparent similarity (they also employ the maximally entangled state and mutually unbiased bases) are not equivalent according to the definition of self-testing (even if we allow for the extra freedom coming from the transposition map).

The first follow-up question that arises from our work is whether the new Bell functional is a self-test in some weaker sense. We conjecture that the maximal violation requires maximal entanglement and mutually unbiased bases, but providing a mathematical formulation of this conjecture is not trivial: for instance, it is clear that not all possible combinations of mutually unbiased bases will give rise to the maximal violation.
%Developing new, weaker formulations of the self-testing problem would be a natural follow-up on this work.

Another interesting direction would be to investigate whether the new Bell functionals can be modified to be maximally violated by other entangled states, keeping at the same time their attractive features, i.e.~analytically computable quantum value and an explicit quantum realisation achieving it. Similarly to how adding a marginal term to the CHSH inequality yields the tilted CHSH inequality~\cite{acin12a}, which has found numerous applications, one might add an analogous local term to the new Bell operator and investigate the consequences. On one hand, one might expect such an inequality to be maximally violated by a non-maximally entangled state of dimension $d$. On the other hand, given the recent non-closure results~\cite{slofstra17a, dykema17a, coladangelo18b}, it is not even guaranteed that the maximal violation can be achieved by finite-dimensional states. Therefore, adding the marginal term could have much more dramatic consequences that in the case of binary outcomes, which makes the problem even more interesting to study.
\section*{Acknowledgements}
We would like to thank Antonio Ac{\'i}n, Joseph Bowles, Roberto Ferrara, Jinhyoung Lee, Giacomo De Palma and Wonmin Son for fruitful discussions. R.~A.~acknowledges the support from the Foundation for Polish Science through the First Team project (First TEAM/2017-4/31) co-financed by the European Union under the European Regional Development Fund. J.~K.~acknowledges support from the POLONEZ programme which has received funding from the European Union's Horizon 2020 research and innovation programme under the Marie Sk{\l}odowska-Curie grant (grant no.~665778), the European Union's Horizon 2020 research and innovation programme under the Marie Sk{\l}odowska-Curie Action ROSETTA (grant no.~749316), the European Research Council (grant no.~337603) and VILLUM FONDEN via the QMATH Centre of Excellence (grant no.~10059). This project has received funding from the European Union's Horizon 2020 research and innovation programme under the Marie-Sk{\l}odowska-Curie grant agreement No 748549. J.~T. acknowledges support from the Alexander von Humboldt Foundation. I. \v{S}., F. B. and A. S. acknowledge the support from Spanish MINECO (QIBEQI FIS2016-80773-P, Severo Ochoa SEV-2015-0522 and a Severo Ochoa PhD fellowship), Fundacio Cellex, Generalitat de Catalunya (SGR875 and CERCA Program), ERC CoG QITBOX and AXA Chair in Quantum Information Science.
%
%
%\clearpage
\begin{widetext}
\appendix
\section{Measurements and observables}
\label{app:measurements-observables}
In this appendix we prove some properties stated in Section~\ref{sec:measurements-observables} and let us start with the following proposition.
\begin{prop}
\label{prop:non-projective-observables}
Let $\{ F_{a} \}_{a}$ be a collection of positive semidefinite operators acting on a finite-dimensional Hilbert space satisfying $\sum_{a} F_{a} = \mathbb{1}$. Then, for arbitrary phases $\phi_{a} \in [0, 2\pi)$ the operator
\begin{equation*}
A := \sum_{a} e^{ \mathbbm{i} \phi_{a} } F_{a}
\end{equation*}
satisfies $A\hc A \leq \mathbb{1}$. Moreover, if the phases are distinct ($\phi_{a} = \phi_{b} \iff a = b$), the operator equality $A\hc A = \mathbb{1}$ holds iff the operators are orthogonal projectors, i.e.~$F_{a} F_{b} = \delta_{ab} F_{a}$ with $\delta_{ab}$ being the standard Kronecker's delta.
\end{prop}
\begin{proof}
Define
\begin{align*}
V &:= \sum_{a} \sqrt{ F_{a} } \otimes \ket{a},\\
U &:= \mathbb{1} \otimes \sum_{a} e^{ \mathbbm{i} \phi_{a} } \ketbraq{a}
\end{align*}
and note that $A = V\hc U V$. The operator $V$ is an isometry ($V\hc V = \mathbb{1}$), which implies that $V V\hc = \Pi$ for some projector $\Pi$ and in particular $V V\hc \leq \mathbb{1}$. Combining this with the fact that $U$ is a unitary immediately implies
\begin{equation*}
A\hc A = V\hc U\hc V V\hc U V \leq V\hc U\hc \cdot \mathbb{1} \cdot U V = V\hc U\hc U V = V\hc V = \mathbb{1}.
\end{equation*}
To prove the second part note that
\begin{equation*}
A\hc A = \sum_{a} F_{a}^{2} + \sum_{a \neq b} e^{\mathbbm{i} ( \phi_{a} - \phi_{b} )} F_{b} F_{a} + e^{- \mathbbm{i} ( \phi_{a} - \phi_{b} )} F_{a} F_{b}.
\end{equation*}
The ``if'' part is clear, so let us focus on the ``only if'' statement. The trace of $A\hc A$ satisfies
\begin{align*}
\tr \big( A\hc A \big) &= \sum_{a} \tr \big( F_{a}^{2} \big) + 2 \sum_{a \neq b} \cos ( \phi_{a} - \phi_{b} ) \tr( F_{a} F_{b} )\\
&\leq \sum_{a} \tr \big( F_{a}^{2} \big) + 2 \sum_{a \neq b} \tr( F_{a} F_{b} ) = \tr \Big[ \big( \sum_{a} F_{a} \big)^{2} \Big] = \tr \mathbb{1}.
\end{align*}
If $A\hc A = \mathbb{1}$, this inequality is tight, which means that for all $a \neq b$ we have
\begin{equation*}
\cos ( \phi_{a} - \phi_{b} ) \tr( F_{a} F_{b} ) = \tr( F_{a} F_{b} ).
\end{equation*}
Since the phases are distinct, we deduce that $\tr ( F_{a} F_{b} ) = 0$, which for positive semidefinite operators implies orthogonality, i.e.~$F_{a} F_{b} = 0$ for $a \neq b$. The fact that the upper bound
\begin{equation*}
A\hc A = \sum_{a} F_{a}^{2} \leq \sum_{a} F_{a} = \mathbb{1}
\end{equation*}
is tight implies that $F_{a}^{2} = F_{a}$ for all $a$.
\end{proof}
Let us now apply these results to the operators $A^{(n)}$. Recall that for a $d$-outcome measurement given by $\{ F_{a} \}_{a = 0}^{d - 1}$ the operator $A^{(n)}$ is defined as
\begin{equation*}
\Andef
\end{equation*}
for $\om{d}$. The first part of Proposition~\ref{prop:non-projective-observables} immediately implies that $[A^{(n)}]\hc A^{(n)} \leq \mathbb{1}$ for all $n$. If the original measurement is projective, i.e.~the measurement operators are orthogonal projectors $F_{a} F_{b} = \delta_{ab} F_{a}$, the measurement can be encoded in an observable $A := A^{(1)}$ (it is easy to verify that $A^{(n)} = A^{n}$). This observable is unitary $A\hc A = A A\hc = \mathbb{1}$ and satisfies $A^{d} = \mathbb{1}$. The second part of Proposition~\ref{prop:non-projective-observables} allows us to deduce that the measurement is projective by looking only at the $A^{(n)}$ operators: the operator $A^{(1)}$ satisfies the condition of the proposition, so if $[A^{(1)}]\hc A^{(1)} = \mathbb{1}$, then the measurement is projective. In fact, the same argument works for $A^{(k)}$ for any integer $k$ which is coprime to $d$.
\section{Optimal observables for $d = 3$}
\label{app:optimal-observables-d=3}
In this appendix we show that the commutation relations derived in Sec.~\ref{sec:self-testing} allow us to reconstruct the optimal observables. Let us start with two technical propositions.
\begin{prop}
\label{prop:commutation-relation}
Let $d$ be an arbitrary integer and let $X$ and $Z$ be the corresponding Heisenberg-\!Weyl operators
\begin{equation*}
X := \sum_{j = 0}^{d - 1} \ketbra{ j + 1 }{j} \nbox{and} Z := \sum_{j = 0}^{d - 1} \omega^{j} \ketbraq{j},
\end{equation*}
where $\ket{d} \equiv \ket{0}$. Let $B_{0}, B_{1}$ be unitaries acting on a finite-dimensional Hilbert space $\sH$ satisfying $B_{0}^{d} = B_{1}^{d} = \mathbb{1}$. Suppose $B_{0}$ and $B_{1}$ satisfy the commutation relation
\begin{equation*}
B_{0} B_{1} = \omega^{q} B_{1} B_{0},
\end{equation*}
where $q$ and $d$ are coprime. Then, $\dim (\sH) =  d \cdot t$ for some integer $t \geq 1$ and there exists a unitary $U : \sH \to \amsbb{C}^{d} \otimes \amsbb{C}^{t}$ such that
\begin{equation}
\label{eq:canonical-choice}
U B_{0} U^{\hc} = Z^{q} \otimes \mathbb{1} \nbox{and} U B_{1} U^{\hc} = X \otimes \mathbb{1}.
\end{equation}
\end{prop}
\begin{proof}
Let $t \in \amsbb{N}$ be the multiplicity of the $\lambda = 1$ eigenvalue of $B_{0}$ and let $\{ \ket{e_{j}^{(0)}} \}_{j = 0}^{t - 1}$ be an orthonormal basis of the corresponding eigenspace. For $k = 1, 2, \ldots, d - 1$ define
\begin{equation*}
\ket{ e_{j}^{(k)} } := B_{1}^{k} \ket{ e_{j}^{(0)} }.
\end{equation*}
The commutation relation implies that for any $k \in \amsbb{N}$
\begin{equation*}
B_{0} B_{1}^{k} = \omega^{qk} B_{1}^{k} B_{0}
\end{equation*}
and therefore
\begin{equation*}
B_{0} \ket{ e_{j}^{(k)} } = \omega^{qk} \ket{ e_{j}^{(k)} }.
%s
\end{equation*}
It is also clear that the vectors $\{ \ket{ e_{j}^{(k)} } \}_{j = 0}^{t - 1}$ span the eigenspace of $B_{0}$ corresponding to the eigenvalue $\omega^{qk}$. Since $q$ and $d$ are coprime, going over $k = 1, 2, \ldots, d - 1$ recovers all the eigenspaces of $B_{0}$. This allows us to deduce that $\dim( \sH ) = d \cdot t$ and that the set $\{ \ket{ e_{j}^{(k)} } \}$ for $j = 0, 1, \ldots, t - 1$ and $k = 0, 1, \ldots, d - 1$ constitutes an orthonormal basis for $\sH$. Writing the two observables in this basis gives
\begin{align*}
B_{0} &= \sum_{j = 0}^{t - 1} \sum_{k = 0}^{d - 1} \omega^{qk} \ketbraq{ e_{j}^{(k)} },\\
B_{1} &= \sum_{j = 0}^{t - 1} \sum_{k = 0}^{d - 1} \ketbra{ e_{j}^{(k + 1)} }{ e_{j}^{(k)} },
\end{align*}
where $\ket{ e_{j}^{(d)} } \equiv \ket{ e_{j}^{(0)} }$. The desired unitary $U : \sH \to \amsbb{C}^{d} \otimes \amsbb{C}^{t}$ is given by
\begin{equation*}
U \ket{ e_{j}^{(k)} } = \ket{k} \ket{g_{j}},
\end{equation*}
where $\{ \ket{k} \}_{k = 0}^{d - 1}$ is the standard basis on $\amsbb{C}^{d}$ and $\{ \ket{ g_{j} } \}_{j = 0}^{t - 1}$ is an arbitrary basis on $\amsbb{C}^{t}$.
\end{proof}
In Eq.~\eqref{eq:canonical-choice} we have chosen the canonical observables to be $Z^{q}$ and $X$, but clearly we can replace them with any observables satisfying the right commutation relation (this is precisely what Prop. \ref{prop:commutation-relation}). For our purposes it is better to make a different choice. For $d = 3$ and $q = 1$ we choose the unitary $U$ such that
\begin{equation}
\label{eq:convenient-choice-1}
U B_{0} U^{\hc} = X \otimes \mathbb{1} \nbox{and} U B_{1} U^{\hc} = X^{2} Z \otimes \mathbb{1},
\end{equation}
whereas for $d = 3$ and $q = 2$ we choose
\begin{equation}
\label{eq:convenient-choice-2}
U B_{0} U^{\hc} = X \otimes \mathbb{1} \nbox{and} U B_{1} U^{\hc} = Z^{2} \otimes \mathbb{1}.
\end{equation}
Let us also mention that this argument could be generalised to infinite-dimensional Hilbert spaces to yield a unitary $U : \sH \to \amsbb{C}^{d} \otimes \sH'$, where both $\sH$ and $\sH'$ are infinite-dimensional.
\begin{prop}
Let $B_{0}, B_{1} \in \cL( \sH_{B} )$ be unitary operators satisfying $B_{0}^{3} = B_{1}^{3} = \mathbb{1}$. If the anticommutator $\{ B_{0}, B_{1} \}$ is unitary, then $\sH_{B} \equiv \sH_{B'} \otimes \sH_{B''}$ for $\sH_{B'} \equiv \amsbb{C}^{3}$ and there exists a unitary $U : \sH_{B} \to \sH_{B'} \otimes \sH_{B''}$ such that
\begin{align*}
U B_{0} U\hc &= X \otimes Q_{1} + X \otimes Q_{2},\\
U B_{1} U\hc &= X^{2} Z \otimes Q_{1} + Z^{2} \otimes Q_{2},
\end{align*}
where $Q_{1}$ and $Q_{2}$ are orthogonal projectors satisfying $Q_{1} + Q_{2} = \mathbb{1}_{B''}$.
\end{prop}
\begin{proof}
The unitarity of the anticommutator reads
\begin{equation*}
\{ B_{0}\hc, B_{1}\hc \} \{ B_{0}, B_{1} \} = \mathbb{1}
\end{equation*}
and is equivalent to
\begin{equation*}
T + T\hc + \mathbb{1} = 0,
\end{equation*}
where $T = B_{0}\hc B_{1}\hc B_{0} B_{1}$. This implies that the eigenvalues of $T$ satisfy $\lambda + \lambda^{*} + 1 = 0$ and since $T$ is unitary, the only possibilities are $\lambda =  \omega$ and $\lambda = \omega^{2}$. Let us now show that the unitaries $B_{0}$ and $B_{1}$ respect the block structure of $T$. We choose a basis in which $T$ reads
\begin{equation*}
T =
\left(
\begin{array}{cc}
\omega \mathbb{1} &\\
& \omega^{2} \mathbb{1}
\end{array}
\right)
\end{equation*}
and write $B_{0}$ in the same basis
\begin{equation*}
B_{0} =
\left(
\begin{array}{cc}
E_{0} & F_{0}\\
F_{1} & E_{1}
\end{array}
\right).
\end{equation*}
The requirement $B_{0}\hc = B_{0}^{2}$ implies
\begin{equation}
\label{eq:b0dagger=b0square}
\left(
\begin{array}{cc}
E_{0}\hc & F_{1}\hc\\
F_{0}\hc & E_{1}\hc
\end{array}
\right) =
\left(
\begin{array}{cc}
E_{0}^{2} + F_{0} F_{1} & E_{0} F_{0} + F_{0} E_{1}\\
F_{1} E_{0} + E_{1} F_{1} & F_{1} F_{0} + E_{1}^{2}
\end{array}
\right).
\end{equation}
Let
\begin{equation}
R := B_{0} T =
\left(
\begin{array}{cc}
\omega E_{0} & \omega^{2} F_{0}\\
\omega F_{1} & \omega^{2} E_{1}
\end{array}
\right).
\end{equation}
Since $R = B_{1}\hc B_{0} B_{1}$, we also have $R\hc = R^{2}$, which in the block form reads
\begin{equation}
\label{eq:rdagger=rsquare}
\left(
\begin{array}{cc}
\omega^{2} E_{0}\hc & \omega^{2} F_{1}\hc\\
\omega F_{0}\hc & \omega E_{1}\hc
\end{array}
\right) =
\left(
\begin{array}{cc}
\omega^{2} E_{0}^{2} + F_{0} F_{1} & E_{0} F_{0} + \omega F_{0} E_{1}\\
\omega^{2} F_{1} E_{0} + E_{1} F_{1} & F_{1} F_{0} + \omega E_{1}^{2}
\end{array}
\right).
\end{equation}
The top-left entries of Eqs.~\eqref{eq:b0dagger=b0square} and \eqref{eq:rdagger=rsquare}, i.e.
\begin{align*}
E_{0}\hc &= E_{0}^{2} + F_{0} F_{1},\\
\omega^{2} E_{0}\hc &= \omega^{2} E_{0}^{2} + F_{0} F_{1}
\end{align*}
immediately imply that $F_{0} F_{1} = 0$. Similarly, the bottom-right entries imply $F_{1} F_{0} = 0$. The bottom-left entries read
\begin{align*}
F_{0}\hc &= F_{1} E_{0} + E_{1} F_{1},\\
\omega F_{0}\hc &= \omega^{2} F_{1} E_{0} + E_{1} F_{1}.
\end{align*}
Right-multiplying both equations by $F_{0}$ yields
\begin{align*}
F_{0}\hc F_{0} &= F_{1} E_{0} F_{0},\\
\omega F_{0}\hc F_{0} &= \omega^{2} F_{1} E_{0} F_{0},
\end{align*}
which immediately implies that $F_{0}\hc F_{0} = 0$ and, therefore, $F_{0} = 0$. Similarly, by looking at the top-right entry we deduce that $F_{1} = 0$. Therefore, $B_{0}$ respects the block structure of $T$. Since the same argument applies to $B_{1}$, when solving the equation $T = B_{0}\hc B_{1}\hc B_{0} B_{1}$ it suffices to solve the two blocks separately. Fortunately, on each block the unitaries $B_{0}$ and $B_{1}$ satisfy a commutation relation covered by Proposition~\ref{prop:commutation-relation}, so we already have the solution. In particular, if we use the canonical observables specified in Eqs.~\eqref{eq:convenient-choice-1} and \eqref{eq:convenient-choice-2} the unitaries $B_{0}$ and $B_{1}$ in the block form are given by
\begin{equation*}
B_{0} =
\left(
\begin{array}{cc}
X \otimes \mathbb{1} &\\
& X \otimes \mathbb{1}
\end{array}
\right)
\nbox{and}
B_{1} =
\left(
\begin{array}{cc}
X^{2} Z \otimes \mathbb{1} &\\
& Z^{2} \otimes \mathbb{1}
\end{array}
\right).
\end{equation*}
The final step is to incorporate the block structure into the tensor product according to the equivalence
\begin{equation*}
( \amsbb{C}^{3} \otimes \amsbb{C}^{t_{1}} ) \oplus ( \amsbb{C}^{3} \otimes \amsbb{C}^{t_{2}} ) \equiv \amsbb{C}^{3} \otimes ( \amsbb{C}^{t_{1}} \oplus \amsbb{C}^{t_{2}} ) \equiv \amsbb{C}^{3} \otimes \amsbb{C}^{t_{1} + t_{2}},
\end{equation*}
which gives rise to the projectors $Q_{1}$ and $Q_{2}$.
\end{proof}
This proposition allows us to prove the results stated in the main text. The commutation relation $B_{2}\hc = - \omega \{ B_{0}, B_{1} \}$ implies that the anticommutator $\{ B_{0}, B_{1} \}$ is unitary, which allows us to determine the exact form of $B_{0}$ and $B_{1}$. Finally, the observable $B_{2}$ can be computed from the same relation.

\section{Quadratic Gauss sums}
\label{app:Gauss}

In this appendix we compute certain quadratic Gauss sums which we use in our considerations in Appendix \ref{app:fazy}.
\begin{obs}
Let $a$ and $b$ be two integers and let $\omega=\mathrm{exp}(2\pi \mathbbm{i}/d)$ with $d$ being a prime number. Then, 
\begin{equation}\label{GaussSum1}
\sum_{i=0}^{d-1}\omega^{a(i^2+bi)}=
\varepsilon_d\sqrt{d}\left(\frac{a}{d}\right)\left\{
\begin{array}{lc}
\omega^{-\frac{1}{4}ab^2},& b \equiv 0 \,\,\mathrm{mod}\,\,2\\[1ex]
\omega^{-\frac{1}{4}a(d-b)^2},& b \equiv 1 \,\,\mathrm{mod}\,\,2
\end{array}
\right.
\end{equation}
\end{obs}
\begin{proof}Assuming first $b$ to be even, we have
the following chain of equalities
\begin{equation}
\sum_{i=0}^{d-1}\omega^{ai^2+abi}=\omega^{-\frac{1}{4}ab^2}\sum_{i=0}^{d-1}\omega^{a\left[i+\frac{b}{2}\right]^2}=\omega^{-\frac{1}{4}ab^2}\sum_{i=\frac{b}{2}}^{d-1+\frac{b}{2}}\omega^{ai^2}=\omega^{-\frac{1}{4}ab^2 }\sum_{i=0}^{d-1}\omega^{ai^2}=\varepsilon_d\sqrt{d}\left(\frac{a}{d}\right)\omega^{-\frac{1 }{4}ab^2},
\end{equation}
where to get the third expression we shifted the summation range by an integer $b/2$, while the third equality is a result of the fact that for prime $d$ this shifting does not change the value of the Gauss sum.

Let us then consider the case of odd $b$. We notice that although $b/2$ is not an integer, $(d-b)/2$ is due to the fact that $d$ is odd and a difference of two odd numbers is 
an even number. Moreover, $\omega^{-ndi}=1$, and therefore we can follow the same reasoning as above, which gives
\begin{equation}
\sum_{i=0}^{d-1}\omega^{a(i^2+bi)}
=\sum_{i=0}^{d-1}\omega^{ai^2}
\omega^{-a(d-b)i}=\omega^{-\frac{1}{4}a(d-b)^2}\sum_{i=0}^{d-1}\omega^{a\left(i-\frac{d-b}{2}\right)^2}=\varepsilon_d\sqrt{d}\left(\frac{a}{d}\right)\omega^{-\frac{a}{4}(d-b)^2}.
\end{equation}
\end{proof}

\begin{obs}Let $a$ be an even integer and $b=c/2$ for some
odd integer $c$ and let $\omega=\mathrm{exp}(2\pi \mathbbm{i}/d)$ with $d$ being a prime number. Then, the following identities hold true
\begin{equation}\label{GaussSum2}
\sum_{i=0}^{d-1}\omega^{a(i^2+bi)}=
\varepsilon_d\sqrt{d}\left(\frac{a}{d}\right)\left\{
\begin{array}{ll}
\omega^{-\frac{1}{4}ab'^2},& b' \equiv 0 \,\,\mathrm{mod}\,\,2\\[1ex]
\omega^{-\frac{1}{4}a(d-b')^2},& b' \equiv 1 \,\,\mathrm{mod}\,\,2,
\end{array}
\right.
\end{equation}
where $b'=b+d/2$.  
\end{obs}
\begin{proof}We could follow the above reasoning, however, 
$b$ is not an integer. To overcome this difficulty, we exploit the fact that $d$ is odd and therefore $b'\equiv b+d/2$ is an integer. Moreover, due to the fact that $a$ is even $\omega^{aid/2}=1$ for any $i$, and consequently, 
\begin{equation}
\sum_{i=0}^{d-1}\omega^{a(i^2+bi)}=\sum_{i=0}^{d-1}\omega^{a\left[i^2+\left(b+\frac{d}{2}\right)i\right]}=\sum_{i=0}^{d-1}\omega^{a(i^2+b'i)}.
\end{equation}
We then obtain (\ref{GaussSum2}) by applying Eq. (\ref{GaussSum1}) to the last term in the above expression, which completes the proof.
\end{proof}

\section{Determining the phases $\lambda_n$}
\label{app:fazy}

Here we will show how the phases $\lambda_n$ appearing in the Bell operator (\ref{eq:Wd-definition}) can be fixed so that the maximal quantum violation of the corresponding Bell inequality is achieved by the maximally entangled state (\ref{eq:maxent}); in other words, we will justify the choice of phases defined in Eqs. (\ref{eq:phases1}) and (\ref{eq:phases12}).
We will also justify the choice of Alice's observables made in the main text. 

To make this section self-contained let us recall the definition of the Bell operator
stated already in Eq. (\ref{eq:Wd-definition}):
\begin{equation}
\label{eq:Wd-definition2}
W_{d} = \frac{ 1 }{d^{2} \sqrt{d}} \sum_{n = 0}^{d - 1} \sum_{j=0}^{d-1} A_{j}^{(n)} \otimes C_{j}^{(n)},
\end{equation}
where
\begin{equation}
\label{appeq:Cjn-definition}
C_j^{(n)}:=\frac{\lambda_n}{\sqrt{d}}\sum_{k=0}^{d-1}\omega^{njk}B_k^{(n)}.
\end{equation}

We begin by deriving the optimal observables of Alice. For this purpose, let us first show 
that for a suitable choice of $\lambda_1$, the operators 
$C_j^{(1)}\equiv C_j$, defined as
\begin{equation}
C_j=\frac{\lambda_1}{\sqrt{d}}\sum_{k=0}^{d-1}\omega^{jk}B_k,
\end{equation}
are proper observables in our scenario, that is, they are unitary and 
have eigenvalues $\omega^{i}$ with $i=0,\ldots,d-1$. In fact, 
it is not difficult to see, with the aid of formula (\ref{GaussSum1}), that for any choice of the phase $\lambda_1$, the operators $C_j$ are indeed unitary. Let us then determine the value of $\lambda_1$ for which the second condition is satisfied too. To this aim, we demand that 
$C_j^d=\mathbbm{1}$
for any $j$, which is equivalent to say that $C_j$ have the required spectrum.

Before exploiting the above condition, let us first obtain a simpler matrix form 
of $C_j$. From (\ref{appeq:Cjn-definition}) and the fact that $B_k=\omega^{k(k+1)}XZ^k$ we have 
\begin{eqnarray}\label{eq:Cj}
C_j&=&
%\frac{\lambda_1}{\sqrt{d}}\sum_{k=0}^{d-1}\omega^{jk}B_k=
\frac{\lambda_1}{\sqrt{d}}\sum_{k=0}^{d-1}\omega^{jk}\omega^{k(k+1)}XZ^k=
\frac{\lambda_1}{\sqrt{d}}\sum_{i=0}^{d-1}\left[\sum_{k=0}^{d-1}\omega^{k^2+k(j+1+i)}\right]\ket{i+1}\!\bra{i}
\equiv
\frac{\lambda_1}{\sqrt{d}}\sum_{i=0}^{d-1}G(i,j,d)\ket{i+1}\!\bra{i},
\end{eqnarray}
where to obtain second equality we have used the explicit matrix form of the $Z$ operator and we have denoted
\begin{equation}
G(i,j,d)=\sum_{k=0}^{d-1}\omega^{k^2+k(i+j+1)}.
\end{equation}
The last sum has already been computed in Appendix
\ref{app:Gauss} and its closed formula is given in 
Eq. (\ref{GaussSum1}).

Now, taking the $d$th power of $C_j$ we obtain
\begin{equation}
C_j^d=\frac{\lambda_1^d}{\sqrt{d}^d}\sum_{i=0}^{d-1}G(i,j,d)\cdot G(i+1,j,d)\cdot\ldots\cdot G(i+d-1,j,d)\ket{i}\!\bra{i}.
\end{equation}
After quite tedious algebra one finds, by virtue of Eq. (\ref{GaussSum1}), that the above product of Gauss sums amounts to 
$G(i,j,d)\cdot\ldots\cdot G(i+d-1,j,d)=\varepsilon_d^d\,d^{d/2}\omega^{-d(d^2-1)/12}$ and thus $C_j^d=\mathbbm{1}$ if 
%
%\begin{equation}
$\lambda_1=\omega^{(d^2-1)/12}/\varepsilon_d$
%\end{equation}
which agrees with Eqs. (\ref{eq:phases1}) and (\ref{eq:phases12}).

Having established that $C_j$ are proper observables in Bell scenario, we define Alice's
observables as $A_j=C_j^{*}$, with the main reason being the fact that in such a case
$A_j\otimes C_j$ is a stabilizing operator of $\ket{\Phi_{AB}}$ for any $j$, that is, 
\begin{equation}\label{stab1}
A_j\otimes C_j\ket{\Phi_{AB}}=\ket{\Phi_{AB}}.
\end{equation}

Let us now determine the phases $\lambda_n$ for $n>1$. 
To this aim we impose the following condition 
\begin{equation}\label{eq:condition}
A_j^{(n)}\otimes C_j^{(n)}\ket{\Phi_{AB}}=\ket{\Phi_{AB}}
\end{equation}
for any $j$ and $n$, which we will use to determine the explicit values of $\lambda_n$. Owing to the well-known property of the maximally entangled state that $X\otimes Y\ket{\Phi_{AB}}=\mathbbm{1}\otimes YX^T\ket{\psi_{AB}}$ for any pair of matrices $X,Y$, the condition (\ref{eq:condition}) can be stated equivalently as
\begin{equation}\label{condition2}
C_j^{(n)}=A_j^{(n)*}=[A_j^{*}]^n=C_j^n.
\end{equation}
In order to exploit this condition, we need to 
find the matrix form of each of its sides. We begin with $C_j^{(n)}$. Using the explicit form of $B_k$ and the fact that 
\begin{equation}
(XZ^k)^n=\sum_{l=0}^{d-1}\omega^{k\frac{n(n-1)}{2}}\omega^{nkl}\ket{l+n}\!\bra{l},
\end{equation}
we can write $C_j^{(n)}$ as 
\begin{equation}
C_j^{(n)}=\frac{\lambda_n}{\sqrt{d}}\sum_{l=0}^{d-1}\left[\sum_{k=0}^{d-1}\omega^{nk^2}\omega^{nk\left(j+l+\frac{n+1}{2}\right)}\right]\ket{l+n}\!\bra{l}
\end{equation}
Using Eqs. (\ref{GaussSum1}) and (\ref{GaussSum2}), 
we finally arrive at 
\begin{equation}\label{Cjnform1}
C_j^{(n)}=\lambda_{n}\,\varepsilon_d\left(\frac{n}{d}\right)\sum_{l=0}^{d-1}|l+n\rangle\!\langle l|\left\{
\begin{array}{lc}
\omega^{-\frac{n}{4}\left(j+l+\frac{n+1}{2}\right)^2}, & j+l+\frac{n+1}{2}\equiv 0\,\,\mathrm{mod}\,\,2\\[1ex]
\omega^{-\frac{n}{4}\left(j+l+\frac{n+1}{2}-d\right)^2}, & j+l+\frac{n+1}{2}\equiv 1\,\,\mathrm{mod}\,\,2,
\end{array}
\right.
\end{equation}
for odd $n$ and 
\begin{equation}\label{Cjnform2}
C_j^{(n)}=\lambda_{n}\,\varepsilon_d\left(\frac{n}{d}\right)\sum_{l=0}^{d-1}|l+n\rangle\!\langle l|\left\{
\begin{array}{cc}
\omega^{-\frac{n}{4}\left(j+l+\frac{n+1+d}{2}\right)^2}, & j+l+\frac{n+1+d}{2}\equiv 0\,\,\mathrm{mod}\,\,2\\[1ex]
\omega^{-\frac{n}{4}\left(j+l+\frac{n+1-d}{2}\right)^2}, & j+l+\frac{n+1+d}{2}\equiv 1\,\,\mathrm{mod}\,\,2.
\end{array}
\right.
\end{equation}
for even $n$.

Let us now move on to $C_j^n$. Using Eq. (\ref{eq:Cj}) we can write
\begin{eqnarray}
C_j^n&=&\frac{\lambda_1^n}{\sqrt{d}^n}\sum_{i=0}^{d-1}G(i,j,d)\cdot G(i+1,j,d)\cdot\ldots\cdot G(i+n-1,j,d)\ket{i+n}\!\bra{i}\nonumber\\
&=&\frac{\lambda_1^n}{\sqrt{d}^n}\left[\sum_{k=0}^{(d-1)/2}G(2k,j,d)\cdot G(2k+1,j,d)\cdot\ldots\cdot G(2k+n-1,j,d)\ket{2k+n}\!\bra{2k}\right.\nonumber\\
&&\hspace{1.3cm}+\left.\sum_{k=0}^{(d-3)/2}G(2k+1,j,d)\cdot G(2k+2,j,d)\cdot\ldots\cdot G(2k+1+n-1,j,d)\ket{2k+1+n}\!\bra{2k+1}\right],
\end{eqnarray}
where to facilitate computation of the above products of Gauss sums we have split the sum into two sums, one over even and one over odd $i$'s. Then, to compute these products we use Eqs. 
(\ref{GaussSum1}) and (\ref{GaussSum2}), dividing our analysis into four cases: 

\begin{itemize}

\item odd $n$, odd $j$,

\begin{equation}\label{OddOdd1}
G(2k,j,d)\ldots G(2k+n-1,j,d)=\varepsilon_d^n\, d^{n/2}\,
%\left\{
%\begin{array}{ll}
\omega^{-\frac{1}{24}\{3d^2(n-1)+n-3d(n-1)(1+2j+4k+n)+n[6(j+2k)(1+j+2k)+3(1+2j+4k)n+2n^2]\}}
%\omega^{-\frac{n}{8}[(-2+d-i-j)(d-i-j)+n(1-d+i+j)]}
%\end{array}
%\right.
\end{equation}
\begin{equation}\label{OddOdd2}
G(2k+1,j,d)\ldots G(2k+1+n-1,j,d)=\varepsilon_d^n\, d^{n/2}\,
%\left\{
%\begin{array}{ll}
\omega^{-\frac{1}{24} \{13 n + 3 d^2 (1 + n) - 3 d (1 + n) (3 + 2 j + 4 k + n) + 
   n [6 (j + 2 k) (3 + j + 2 k) + 3 (3 + 2 j + 4 k) n + 2 n^2]\}}\\
%\omega^{-\frac{d}{24}[3(i+j)^2-3d(i+j)+d^2-1]}\omega^{-\frac{n}{12}}\omega^{-\frac{n^3}{24}}
%\omega^{-\frac{n}{8}[(-2+d-i-j)(d-i-j)+n(1-d+i+j)]}
%\end{array}
%\right.
\end{equation}

\item odd $n$, even $j$,

\begin{equation}\label{OddEven1}
G(2k,j,d)\ldots G(2k+n-1,j,d)=\varepsilon_d^n\, d^{n/2}\,
%\left\{
%\begin{array}{ll}
\omega^{-\frac{1}{24}\{n+3d^2(n+1)-3d(n+1)(1+2j+4k+n)+n[6(j+2k)(j+1+2k)+3(1+2j+4k)n+2n^2]\}}
%\omega^{-\frac{n}{8}[(-2+d-i-j)(d-i-j)+n(1-d+i+j)]}
%\end{array}
%\right.
\end{equation}
\begin{equation}\label{OddEven2}
G(2k+1,j,d)\ldots G(2k+1+n-1,j,d)=\varepsilon_d^n\, d^{n/2}\,
%\left\{
%\begin{array}{ll}
\omega^{-\frac{1}{24}\{3d^2(n-1)-3d(n-1)(3+4k+n)+n[13+24k^2+12k(3+n)+n(9+2n)]\}}\\
%\omega^{-\frac{d}{24}[3(i+j)^2-3d(i+j)+d^2-1]}\omega^{-\frac{n}{12}}\omega^{-\frac{n^3}{24}}
%\omega^{-\frac{n}{8}[(-2+d-i-j)(d-i-j)+n(1-d+i+j)]}
%\end{array}
%\right.
\end{equation}

\item even $n$, odd $j$,

\begin{equation}\label{EvenOdd1}
G(2k,j,d)\ldots G(2k+n-1,j,d)=\varepsilon_d^n\, d^{n/2}\,
%\left\{
%\begin{array}{ll}
\omega^{-\frac{n}{24}\{1+3d^2+6j+12k+6(j+2k)^2+3n+6(j+2k)n+2n^2-3d(2+2j+4k+n)\}}
%\omega^{-\frac{n}{8}[(-2+d-i-j)(d-i-j)+n(1-d+i+j)]}
%\end{array}
%\right.
\end{equation}
\begin{equation}\label{EvenOdd2}
G(2k+1,j,d)\ldots G(2k+1+n-1,j,d)=\varepsilon_d^n\, d^{n/2}\,
%\left\{
%\begin{array}{ll}
\omega^{-\frac{n}{24}\{13+3d^2+18j+36k+6(j+2k)^2+9n+6(j+2k)n+2n^2-3d(2+2j+4k+n)\}}\\
%\omega^{-\frac{d}{24}[3(i+j)^2-3d(i+j)+d^2-1]}\omega^{-\frac{n}{12}}\omega^{-\frac{n^3}{24}}
%\omega^{-\frac{n}{8}[(-2+d-i-j)(d-i-j)+n(1-d+i+j)]}
%\end{array}
%\right.
\end{equation}

\item Even $n$, even $j$,

\begin{equation}\label{EvenEven1}
G(2k,j,d)\ldots G(2k+n-1,j,d)=\varepsilon_d^n\, d^{n/2}\,
%\left\{
%\begin{array}{ll}
\omega^{-\frac{n}{24}\{1+3d^2+6j+12k+6(j+2k)^2+3n+6(j+2k)n+2n^2-3d(2j+4k+n)\}}
%\omega^{-\frac{n}{8}[(-2+d-i-j)(d-i-j)+n(1-d+i+j)]}
%\end{array}
%\right.
\end{equation}
\begin{equation}\label{EvenEven2}
G(2k+1,j,d)\ldots G(2k+1+n-1,j,d)=\varepsilon_d^n\, d^{n/2}\,
%\left\{
%\begin{array}{ll}
\omega^{-\frac{n}{24}\{13+3d^2+18j+36k+6(j+2k)^2+9n+6(j+2k)n+2n^2-3d(4+2j+4k+n)\}}\\
%\omega^{-\frac{d}{24}[3(i+j)^2-3d(i+j)+d^2-1]}\omega^{-\frac{n}{12}}\omega^{-\frac{n^3}{24}}
%\omega^{-\frac{n}{8}[(-2+d-i-j)(d-i-j)+n(1-d+i+j)]}
%\end{array}
%\right.
\end{equation}

\end{itemize}

Having determined both sides of Eq. (\ref{condition2}), we can now compare them. 
Traditionally, we will consider the cases of odd and even $n$ separately. \\

\noindent\textbf{Odd $n$.} Let us assume that $n\mod 4\equiv 1$, i.e., $n=4m+1$ with $m\in\mathbbm{N}$. Let also $j$ be odd. Then, to determine $\lambda_n$ we compare Eq. (\ref{OddOdd1}) and the first formula in Eq. (\ref{Cjnform1}) with $l=2k$, which, after some algebra, gives us
\begin{equation}\label{lambda1}
\lambda_{n}=\frac{1}{\varepsilon_d \left(\frac{n}{d}\right)}\omega^{-\frac{1}{48}[n(n^2+3)+2d^2(-5n+3)]}.
\end{equation}
We then, check that if we compare Eq. (\ref{OddOdd2}) with the second formula in Eq. (\ref{Cjnform1}) with $l=2k+1$, we obtain exactly the same phases.

On the other hand, if we assume that $n\mod 4\equiv 3$, i.e., $n=4m+3$ with $m\in\mathbbm{N}$
and also that $j$ is odd, we compare Eq. (\ref{OddOdd2}) with the first formula in 
Eq. (\ref{Cjnform1}) with $l=2k+1$, which leads us to
\begin{equation}\label{lambda2}
\lambda_{n}=\frac{1}{\varepsilon_d \left(\frac{n}{d}\right)}\omega^{-\frac{1}{48}[n(n^2+3)+2d^2(n+3)]}.
\end{equation}
The same formula is obtained when comparing Eq. (\ref{OddOdd1}) with the second formula in 
Eq. (\ref{Cjnform1}) with $l=2k$.\\

On can check that exactly the same formulas are obtained in the case of even $j$.\\

\noindent\textbf{Even $n$.} Assume first that $j$ is odd and $n\mod 4\equiv 0$, meaning that $n=4m$ with $m\in\mathbbm{N}$. Let us also assume that $d=4p+1$ with $p\in\mathbbm{N}$. Then, comparison of Eq. (\ref{EvenOdd1}) with the first formula in Eq. (\ref{Cjnform2}) for $l=2k$ as well as Eq. (\ref{EvenOdd2}) with the second formula in Eq. (\ref{Cjnform2}) for $l=2k+1$, gives
\begin{equation}\label{lambda3}
\lambda_{n}=\frac{1}{\varepsilon_d \left(\frac{n}{d}\right)}\omega^{-\frac{n}{48}[n^2-d(d-6)+3]}.
\end{equation}
If we then assume that $d=4p+3$ with $p\in\mathbbm{N}$ and compare Eq. (\ref{EvenOdd1})
with the second formula in (\ref{Cjnform2}) with $l=2k$, we obtain
\begin{equation}\label{lambda4}
\lambda_{n}=\frac{1}{\varepsilon_d\left(\frac{n}{d}\right)}\omega^{-\frac{n}{48}[n^2-d(d+6)+3]}.
\end{equation}
Comparison of Eq. (\ref{EvenOdd2}) with 
the first formula in Eq. (\ref{Cjnform2}) with $l=2k+1$, leads to the same formula.

Let us finally consider the case of $n=4m+2$ for any $m\in\mathbbm{N}$. Then, 
as before we consider two cases $d=4p+1$ and $d=4p+3$ with $p\in\mathbbm{N}$.
In the first case, we use Eq. (\ref{EvenOdd1}) and the second formula in Eq. (\ref{Cjnform2})
to get (\ref{lambda4}). As the same time, the same formula for $\lambda_n$ is obtained from 
Eq. (\ref{EvenOdd2}) and the first formula in Eq. (\ref{Cjnform2}) with $l=2k$. 

Then, in the second case, i.e., $d=4p+3$, we exploit
Eq. (\ref{EvenOdd1}) with the first formula in Eq. (\ref{Cjnform2}), which leads us to  
$\lambda_n$ given in Eq. (\ref{lambda3}). At the same time, comparison of Eq. (\ref{EvenOdd2}) with the second formula in Eq. (\ref{Cjnform2}) with $l=2k+1$ gives exactly the same formula.

One then checks that the same phases are obtained under the assumption that $j$ is even.\\

\noindent\textbf{Summary.} To summarize, depending on the value of $n$ we use the following $\lambda$'s:
\begin{itemize}
\item $n=4m$. We use Eq. (\ref{lambda3}) for $d=4p+3$ and Eq. (\ref{lambda4}) for
$d=4p+1$.

\item $n=4m+1$. We use Eq. (\ref{lambda2}) for any $d$.

\item $n=4m+2$. We use Eq. (\ref{lambda3}) for $d=4p+1$ and Eq. (\ref{lambda4}) for $d=4p+3$.

\item $n=4m+3$. We use Eq. (\ref{lambda1}) irrespectively of the dimension. 
\end{itemize}
\end{widetext}
%
%\nocite{apsrev41Control}
%\bibliographystyle{alphaarxiv}
%\bibliography{library}

\begin{thebibliography}{AFDF{\etalchar{+}}18}

\bibitem[ABB{\etalchar{+}}17]{andersson17a}
O.~Andersson, P.~Badzi\k{a}g, I.~Bengtsson, I.~Dumitru, and A.~Cabello.
\newblock {Self-testing properties of Gisin's elegant Bell inequality}.
\newblock {\em Phys. Rev. A}, 96: 032119, 2017.
\newblock \\
  \texttt{DOI:\,\href{http://dx.doi.org/10.1103/PhysRevA.96.032119}{10.1103/PhysRevA.96.032119}}.

\bibitem[ABG{\etalchar{+}}07]{acin07a}
A.~Ac{\'i}n, N.~Brunner, N.~Gisin, S.~Massar, S.~Pironio, and V.~Scarani.
\newblock {Device-independent security of quantum cryptography against
  collective attacks}.
\newblock {\em Phys. Rev. Lett.}, 98: 230501, 2007.
\newblock \\
  \texttt{DOI:\,\href{http://dx.doi.org/10.1103/PhysRevLett.98.230501}{10.1103/PhysRevLett.98.230501}}.

\bibitem[AFDF{\etalchar{+}}18]{arnonfriedman18a}
R.~Arnon-Friedman, F.~Dupuis, O.~Fawzi, R.~Renner, and T.~Vidick.
\newblock {Practical device-independent quantum cryptography via entropy
  accumulation}.
\newblock {\em Nat. Commun.}, 9: 459, 2018.
\newblock \\
  \texttt{DOI:\,\href{http://dx.doi.org/10.1038/s41467-017-02307-4}{10.1038/s41467-017-02307-4}}.

\bibitem[AGM06]{acin06a}
A.~Ac{\'i}n, N.~Gisin, and L.~Masanes.
\newblock {From Bell's theorem to secure quantum key distribution}.
\newblock {\em Phys. Rev. Lett.}, 97: 120405, 2006.
\newblock \\
  \texttt{DOI:\,\href{http://dx.doi.org/10.1103/PhysRevLett.97.120405}{10.1103/PhysRevLett.97.120405}}.

\bibitem[AMP12]{acin12a}
A.~Ac{\'i}n, S.~Massar, and S.~Pironio.
\newblock {Randomness versus nonlocality and entanglement}.
\newblock {\em Phys. Rev. Lett.}, 108: 100402, 2012.
\newblock \\
  \texttt{DOI:\,\href{http://dx.doi.org/10.1103/PhysRevLett.108.100402}{10.1103/PhysRevLett.108.100402}}.

\bibitem[BBRV02]{bandyopadhyay02a}
S.~Bandyopadhyay, P.~O. Boykin, V.~Roychowdhury, and F.~Vatan.
\newblock {A new proof for the existence of mutually unbiased bases}.
\newblock {\em Algorithmica}, 34: 512, 2002.
\newblock \\
  \texttt{DOI:\,\href{http://dx.doi.org/10.1007/s00453-002-0980-7}{10.1007/s00453-002-0980-7}}.

\bibitem[BCP{\etalchar{+}}14]{brunner14a}
N.~Brunner, D.~Cavalcanti, S.~Pironio, V.~Scarani, and S.~Wehner.
\newblock {Bell nonlocality}.
\newblock {\em Rev. Mod. Phys.}, 86: 419, 2014.
\newblock \\
  \texttt{DOI:\,\href{http://dx.doi.org/10.1103/RevModPhys.86.419}{10.1103/RevModPhys.86.419}}.

\bibitem[Bel64]{bell64a}
J.~S. Bell.
\newblock {On the Einstein-Podolsky-Rosen paradox}.
\newblock {\em Physics}, 1: 195, 1964.

\bibitem[BHK05]{barrett05b}
J.~Barrett, L.~Hardy, and A.~Kent.
\newblock {No signaling and quantum key distribution}.
\newblock {\em Phys. Rev. Lett.}, 95: 010503, 2005.
\newblock \\
  \texttt{DOI:\,\href{http://dx.doi.org/10.1103/PhysRevLett.95.010503}{10.1103/PhysRevLett.95.010503}}.

\bibitem[BKP06]{barrett06a}
J.~Barrett, A.~Kent, and S.~Pironio.
\newblock {Maximally nonlocal and monogamous quantum correlations}.
\newblock {\em Phys. Rev. Lett.}, 97: 170409, 2006.
\newblock \\
  \texttt{DOI:\,\href{http://dx.doi.org/10.1103/PhysRevLett.97.170409}{10.1103/PhysRevLett.97.170409}}.

\bibitem[BLM{\etalchar{+}}09]{bardyn09a}
C.-E. Bardyn, T.~C.~H. Liew, S.~Massar, M.~McKague, and V.~Scarani.
\newblock {Device independent state estimation based on Bell's inequalities}.
\newblock {\em Phys. Rev. A}, 80: 062327, 2009.
\newblock \\
  \texttt{DOI:\,\href{http://dx.doi.org/10.1103/PhysRevA.80.062327}{10.1103/PhysRevA.80.062327}}.

\bibitem[BM05]{buhrman05a}
H.~Buhrman and S.~Massar.
\newblock {Causality and Tsirelson's bounds}.
\newblock {\em Phys. Rev. A}, 72: 052103, 2005.
\newblock \\
  \texttt{DOI:\,\href{http://dx.doi.org/10.1103/PhysRevA.72.052103}{10.1103/PhysRevA.72.052103}}.

\bibitem[BP15]{bamps15a}
C.~Bamps and S.~Pironio.
\newblock {Sum-of-squares decompositions for a family of
  Clauser-Horne-Shimony-Holt-like inequalities and their application to
  self-testing}.
\newblock {\em Phys. Rev. A}, 91: 052111, 2015.
\newblock \\
  \texttt{DOI:\,\href{http://dx.doi.org/10.1103/PhysRevA.91.052111}{10.1103/PhysRevA.91.052111}}.

\bibitem[BPA{\etalchar{+}}08]{brunner08a}
N.~Brunner, S.~Pironio, A.~Ac{\'i}n, N.~Gisin, A.~A. M{\'e}thot, and
  V.~Scarani.
\newblock {Testing the dimension of Hilbert spaces}.
\newblock {\em Phys. Rev. Lett.}, 100: 210503, 2008.
\newblock \\
  \texttt{DOI:\,\href{http://dx.doi.org/10.1103/PhysRevLett.100.210503}{10.1103/PhysRevLett.100.210503}}.

\bibitem[BPPP14]{bouda14a}
J.~Bouda, M.~Paw{\l}owski, M.~Pivoluska, and M.~Plesch.
\newblock {Device-independent randomness extraction from an arbitrarily weak
  min-entropy source}.
\newblock {\em Phys. Rev. A}, 90: 032313, 2014.
\newblock \\
  \texttt{DOI:\,\href{http://dx.doi.org/10.1103/PhysRevA.90.032313}{10.1103/PhysRevA.90.032313}}.

\bibitem[BS15]{bavarian15a}
M.~Bavarian and P.~W. Shor.
\newblock {Information causality, Szemer{\'e}di-Trotter and algebraic variants
  of CHSH}.
\newblock {\em Proc. Conference on Innovations in Theoretical Computer
  Science}, 2015.
\newblock \\
  \texttt{DOI:\,\href{http://dx.doi.org/10.1145/2688073.2688112}{10.1145/2688073.2688112}}.

\bibitem[CBLC16]{chen16a}
S.-L. Chen, C.~Budroni, Y.-C. Liang, and Y.-N. Chen.
\newblock {Natural framework for device-independent quantification of quantum
  steerability, measurement incompatibility, and self-testing}.
\newblock {\em Phys. Rev. Lett.}, 116: 240401, 2016.
\newblock \\
  \texttt{DOI:\,\href{http://dx.doi.org/10.1103/PhysRevLett.116.240401}{10.1103/PhysRevLett.116.240401}}.

\bibitem[CGL{\etalchar{+}}02]{collins02a}
D.~Collins, N.~Gisin, N.~Linden, S.~Massar, and S.~Popescu.
\newblock {Bell inequalities for arbitrarily high-dimensional systems}.
\newblock {\em Phys. Rev. Lett.}, 88: 040404, 2002.
\newblock \\
  \texttt{DOI:\,\href{http://dx.doi.org/10.1103/PhysRevLett.88.040404}{10.1103/PhysRevLett.88.040404}}.

\bibitem[CGS17]{coladangelo17a}
A.~Coladangelo, K.~T. Goh, and V.~Scarani.
\newblock {All pure bipartite entangled states can be self-tested}.
\newblock {\em Nat. Commun.}, 8: 15485, 2017.
\newblock \\
  \texttt{DOI:\,\href{http://dx.doi.org/10.1038/ncomms15485}{10.1038/ncomms15485}}.

\bibitem[CHSH69]{clauser69a}
J.~F. Clauser, M.~A. Horne, A.~Shimony, and R.~A. Holt.
\newblock {Proposed experiment to test local hidden-variable theories}.
\newblock {\em Phys. Rev. Lett.}, 23: 880, 1969.
\newblock \\
  \texttt{DOI:\,\href{http://dx.doi.org/10.1103/PhysRevLett.23.880}{10.1103/PhysRevLett.23.880}}.

\bibitem[CK11]{colbeck11a}
R.~Colbeck and A.~Kent.
\newblock {Private randomness expansion with untrusted devices}.
\newblock {\em J. Phys. A: Math. Theor.}, 44: 095305, 2011.
\newblock \\
  \texttt{DOI:\,\href{http://dx.doi.org/10.1088/1751-8113/44/9/095305}{10.1088/1751-8113/44/9/095305}}.

\bibitem[Col06]{colbeck06a}
R.~Colbeck.
\newblock {\em {Quantum and relativistic protocols for secure multi-party
  computation}}.
\newblock PhD thesis, University of Cambridge, 2006.

\bibitem[Col18]{coladangelo18a}
A.~Coladangelo.
\newblock {Generalization of the Clauser-Horne-Shimony-Holt inequality
  self-testing maximally entangled states of any local dimension}.
\newblock {\em Phys. Rev. A}, 98: 052115, 2018.
\newblock \\
  \texttt{DOI:\,\href{http://dx.doi.org/10.1103/PhysRevA.98.052115}{10.1103/PhysRevA.98.052115}}.

\bibitem[CS16]{cavalcanti16a}
D.~Cavalcanti and P.~Skrzypczyk.
\newblock {Quantitative relations between measurement incompatibility, quantum
  steering, and nonlocality}.
\newblock {\em Phys. Rev. A}, 93: 052112, 2016.
\newblock \\
  \texttt{DOI:\,\href{http://dx.doi.org/10.1103/PhysRevA.93.052112}{10.1103/PhysRevA.93.052112}}.

\bibitem[CS17]{coladangelo17c}
A.~Coladangelo and J.~Stark.
\newblock {Robust self-testing for linear constraint system games}.
\newblock 2017.

\bibitem[CS18]{coladangelo18b}
A.~Coladangelo and J.~Stark.
\newblock {Unconditional separation of finite and infinite-dimensional quantum
  correlations}.
\newblock 2018.

\bibitem[DPP17]{dykema17a}
K.~Dykema, V.~I. Paulsen, and J.~Prakash.
\newblock {Non-closure of the set of quantum correlations via graphs}.
\newblock 2017.

\bibitem[dV15]{devicente15a}
J.~I. de~Vicente.
\newblock {Simple conditions constraining the set of quantum correlations}.
\newblock {\em Phys. Rev. A}, 92: 032103, 2015.
\newblock \\
  \texttt{DOI:\,\href{http://dx.doi.org/10.1103/PhysRevA.92.032103}{10.1103/PhysRevA.92.032103}}.

\bibitem[EPR35]{einstein35a}
A.~Einstein, B.~Podolsky, and N.~Rosen.
\newblock {Can quantum-mechanical description of physical reality be considered
  complete?}
\newblock {\em Phys. Rev.}, 47: 777, 1935.
\newblock \\
  \texttt{DOI:\,\href{http://dx.doi.org/10.1103/PhysRev.47.777}{10.1103/PhysRev.47.777}}.

\bibitem[GKW{\etalchar{+}}18]{goh18a}
K.~T. Goh, J.~Kaniewski, E.~Wolfe, T.~V{\'e}rtesi, X.~Wu, Y.~Cai, Y.-C. Liang,
  and V.~Scarani.
\newblock {Geometry of the set of quantum correlations}.
\newblock {\em Phys. Rev. A}, 97: 022104, 2018.
\newblock \\
  \texttt{DOI:\,\href{http://dx.doi.org/10.1103/PhysRevA.97.022104}{10.1103/PhysRevA.97.022104}}.

\bibitem[JLL{\etalchar{+}}08]{ji08a}
S.-W. Ji, J.~Lee, J.~Lim, K.~Nagata, and H.-W. Lee.
\newblock {Multisetting Bell inequality for qudits}.
\newblock {\em Phys. Rev. A}, 78: 052103, 2008.
\newblock \\
  \texttt{DOI:\,\href{http://dx.doi.org/10.1103/PhysRevA.78.052103}{10.1103/PhysRevA.78.052103}}.

\bibitem[Kan17]{kaniewski17a}
J.~Kaniewski.
\newblock {Self-testing of binary observables based on commutation}.
\newblock {\em Phys. Rev. A}, 95: 062323, 2017.
\newblock \\
  \texttt{DOI:\,\href{http://dx.doi.org/10.1103/PhysRevA.95.062323}{10.1103/PhysRevA.95.062323}}.

\bibitem[KW16]{kaniewski16a}
J.~Kaniewski and S.~Wehner.
\newblock {Device-independent two-party cryptography secure against sequential
  attacks}.
\newblock {\em New J. Phys.}, 18: 055004, 2016.
\newblock \\
  \texttt{DOI:\,\href{http://dx.doi.org/10.1088/1367-2630/18/5/055004}{10.1088/1367-2630/18/5/055004}}.

\bibitem[LLD09]{liang09a}
Y.-C. Liang, C.-W. Lim, and D.-L. Deng.
\newblock {Reexamination of a multisetting Bell inequality for qudits}.
\newblock {\em Phys. Rev. A}, 80: 052116, 2009.
\newblock \\
  \texttt{DOI:\,\href{http://dx.doi.org/10.1103/PhysRevA.80.052116}{10.1103/PhysRevA.80.052116}}.

\bibitem[LRY{\etalchar{+}}10]{lim10a}
J.~Lim, J.~Ryu, S.~Yoo, C.~Lee, J.~Bang, and J.~Lee.
\newblock {Genuinely high-dimensional nonlocality optimized by complementary
  measurements}.
\newblock {\em New J. Phys.}, 12: 103012, 2010.
\newblock \\
  \texttt{DOI:\,\href{http://dx.doi.org/10.1088/1367-2630/12/10/103012}{10.1088/1367-2630/12/10/103012}}.

\bibitem[MBL{\etalchar{+}}13]{moroder13a}
T.~Moroder, J.-D. Bancal, Y.-C. Liang, M.~Hofmann, and O.~G{\"u}hne.
\newblock {Device-independent entanglement quantification and related
  applications}.
\newblock {\em Phys. Rev. Lett.}, 111: 030501, 2013.
\newblock \\
  \texttt{DOI:\,\href{http://dx.doi.org/10.1103/PhysRevLett.111.030501}{10.1103/PhysRevLett.111.030501}}.

\bibitem[McK14]{mckague14a}
M.~McKague.
\newblock {Self-testing graph states}.
\newblock {\em Theory of Quantum Computation, Communication, and Cryptography.
  TQC 2011. Lecture Notes in Computer Science}, 6745: 104, 2014.
\newblock \\
  \texttt{DOI:\,\href{http://dx.doi.org/10.1007/978-3-642-54429-3\_7}{10.1007/978-3-642-54429-3\_7}}.

\bibitem[MM11]{mckague11a}
M.~McKague and M.~Mosca.
\newblock {Generalized self-testing and the security of the 6-state protocol}.
\newblock {\em Theory of Quantum Computation, Communication, and Cryptography.
  TQC 2010. Lecture Notes in Computer Science}, 6519: 113, 2011.
\newblock \\
  \texttt{DOI:\,\href{http://dx.doi.org/10.1007/978-3-642-18073-6\_10}{10.1007/978-3-642-18073-6\_10}}.

\bibitem[MS16]{miller16b}
C.~A. Miller and Y.~Shi.
\newblock {Robust protocols for securely expanding randomness and distributing
  keys using untrusted quantum devices}.
\newblock {\em J. ACM}, 63: 33, 2016.
\newblock \\
  \texttt{DOI:\,\href{http://dx.doi.org/10.1145/2885493}{10.1145/2885493}}.

\bibitem[MY98]{mayers98a}
D.~Mayers and A.~Yao.
\newblock {Quantum cryptography with imperfect apparatus}.
\newblock {\em Proceedings 39th Annual Symposium on Foundations of Computer
  Science}, 1998.
\newblock \\
  \texttt{DOI:\,\href{http://dx.doi.org/10.1109/SFCS.1998.743501}{10.1109/SFCS.1998.743501}}.

\bibitem[MY04]{mayers04a}
D.~Mayers and A.~Yao.
\newblock {Self testing quantum apparatus}.
\newblock {\em Quant. Inf. Comp.}, 4: 273, 2004.

\bibitem[MYS12]{mckague12a}
M.~McKague, T.~H. Yang, and V.~Scarani.
\newblock {Robust self-testing of the singlet}.
\newblock {\em J. Phys. A: Math. Theor.}, 45: 455304, 2012.
\newblock \\
  \texttt{DOI:\,\href{http://dx.doi.org/10.1088/1751-8113/45/45/455304}{10.1088/1751-8113/45/45/455304}}.

\bibitem[PAM{\etalchar{+}}10]{pironio10a}
S.~Pironio, A.~Ac{\'i}n, S.~Massar, A.~Boyer de~la Giroday, D.~N. Matsukevich,
  P.~Maunz, S.~Olmschenk, D.~Hayes, L.~Luo, T.~A. Manning, and C.~Monroe.
\newblock {Random numbers certified by Bell's theorem}.
\newblock {\em Nature}, 464: 1021, 2010.
\newblock \\
  \texttt{DOI:\,\href{http://dx.doi.org/10.1038/nature09008}{10.1038/nature09008}}.

\bibitem[PR92]{popescu92a}
S.~Popescu and D.~Rohrlich.
\newblock {Which states violate Bell's inequality maximally?}
\newblock {\em Phys. Lett. A}, 169: 411, 1992.
\newblock \\
  \texttt{DOI:\,\href{http://dx.doi.org/10.1016/0375-9601(92)90819-8}{10.1016/0375-9601(92)90819-8}}.

\bibitem[PR94]{popescu94a}
S.~Popescu and D.~Rohrlich.
\newblock {Quantum nonlocality as an axiom}.
\newblock {\em Found. Phys.}, 24: 379, 1994.
\newblock \\
  \texttt{DOI:\,\href{http://dx.doi.org/10.1007/BF02058098}{10.1007/BF02058098}}.

\bibitem[RMW16]{ribeiro16a}
J.~Ribeiro, G.~Murta, and S.~Wehner.
\newblock {Fully general device-independence for two-party cryptography and
  position verification}.
\newblock 2016.

\bibitem[RMW18]{ribeiro18b}
J.~Ribeiro, G.~Murta, and S.~Wehner.
\newblock {Fully device-independent conference key agreement}.
\newblock {\em Phys. Rev. A}, 97: 022307, 2018.
\newblock \\
  \texttt{DOI:\,\href{http://dx.doi.org/10.1103/PhysRevA.97.022307}{10.1103/PhysRevA.97.022307}}.

\bibitem[RTK{\etalchar{+}}18]{ribeiro18a}
J.~Ribeiro, L.~P. Thinh, J.~Kaniewski, J.~Helsen, and S.~Wehner.
\newblock {Device independence for two-party cryptography and position
  verification with memoryless devices}.
\newblock {\em Phys. Rev. A}, 97: 062307, 2018.
\newblock \\
  \texttt{DOI:\,\href{http://dx.doi.org/10.1103/PhysRevA.97.062307}{10.1103/PhysRevA.97.062307}}.

\bibitem[RUV13]{reichardt13a}
B.~W. Reichardt, F.~Unger, and U.~Vazirani.
\newblock {Classical command of quantum systems}.
\newblock {\em Nature}, 496: 456, 2013.
\newblock \\
  \texttt{DOI:\,\href{http://dx.doi.org/10.1038/nature12035}{10.1038/nature12035}}.

\bibitem[{\v{S}}ASA16]{supic16a}
I.~{\v{S}}upi{\'c}, R.~Augusiak, A.~Salavrakos, and A.~Ac{\'i}n.
\newblock {Self-testing protocols based on the chained Bell inequalities}.
\newblock {\em New J. Phys.}, 18: 035013, 2016.
\newblock \\
  \texttt{DOI:\,\href{http://dx.doi.org/10.1088/1367-2630/18/3/035013}{10.1088/1367-2630/18/3/035013}}.

\bibitem[SAT{\etalchar{+}}17]{salavrakos17a}
A.~Salavrakos, R.~Augusiak, J.~Tura, P.~Wittek, A.~Ac{\'i}n, and S.~Pironio.
\newblock {Bell inequalities tailored to maximally entangled states}.
\newblock {\em Phys. Rev. Lett.}, 119: 040402, 2017.
\newblock \\
  \texttt{DOI:\,\href{http://dx.doi.org/10.1103/PhysRevLett.119.040402}{10.1103/PhysRevLett.119.040402}}.

\bibitem[SCA{\etalchar{+}}11]{silman11a}
J.~Silman, A.~Chailloux, N.~Aharon, I.~Kerenidis, S.~Pironio, and S.~Massar.
\newblock {Fully distrustful quantum bit commitment and coin flipping}.
\newblock {\em Phys. Rev. Lett.}, 106: 220501, 2011.
\newblock \\
  \texttt{DOI:\,\href{http://dx.doi.org/10.1103/PhysRevLett.106.220501}{10.1103/PhysRevLett.106.220501}}.

\bibitem[{\v{S}}CAA18]{supic18a}
I.~{\v{S}}upi{\'c}, A.~Coladangelo, R.~Augusiak, and A.~Ac{\'i}n.
\newblock {Self-testing multipartite entangled states through projections onto
  two systems}.
\newblock {\em New J. Phys.}, 20: 083041, 2018.
\newblock \\
  \texttt{DOI:\,\href{http://dx.doi.org/10.1088/1367-2630/aad89b}{10.1088/1367-2630/aad89b}}.

\bibitem[SLK06]{son06a}
W.~Son, J.~Lee, and M.~S. Kim.
\newblock {Generic Bell inequalities for multipartite arbitrary dimensional
  systems}.
\newblock {\em Phys. Rev. Lett.}, 96: 060406, 2006.
\newblock \\
  \texttt{DOI:\,\href{http://dx.doi.org/10.1103/PhysRevLett.96.060406}{10.1103/PhysRevLett.96.060406}}.

\bibitem[Slo17]{slofstra17a}
W.~Slofstra.
\newblock {The set of quantum correlations is not closed}.
\newblock 2017.

\bibitem[SW87]{summers87a}
S.~J. Summers and R.~F. Werner.
\newblock {Maximal violation of Bell's inequalities is generic in quantum field
  theory}.
\newblock {\em Commun. Math. Phys.}, 110: 247, 1987.
\newblock \\
  \texttt{DOI:\,\href{http://dx.doi.org/10.1007/BF01207366}{10.1007/BF01207366}}.

\bibitem[Tsi87]{tsirelson87a}
B.~S. Tsirelson.
\newblock {Quantum analogues of the Bell inequalities. The case of two
  spatially separated domains}.
\newblock {\em J. Soviet Math.}, 36: 557, 1987.
\newblock \\
  \texttt{DOI:\,\href{http://dx.doi.org/10.1007/BF01663472}{10.1007/BF01663472}}.

\bibitem[Tsi93]{tsirelson93a}
B.~S. Tsirelson.
\newblock {Some results and problems on quantum Bell-type inequalities}.
\newblock {\em Hadronic J. Suppl.}, 8: 329, 1993.

\bibitem[VV12]{vazirani12a}
U.~Vazirani and T.~Vidick.
\newblock {Certifiable quantum dice: or, true random number generation secure
  against quantum adversaries}.
\newblock {\em Proceedings 44th Annual ACM Symposium on Theory of Computing},
  2012.
\newblock \\
  \texttt{DOI:\,\href{http://dx.doi.org/10.1145/2213977.2213984}{10.1145/2213977.2213984}}.

\bibitem[VV14]{vazirani14a}
U.~Vazirani and T.~Vidick.
\newblock {Fully device-independent quantum key distribution}.
\newblock {\em Phys. Rev. Lett.}, 113: 140501, 2014.
\newblock \\
  \texttt{DOI:\,\href{http://dx.doi.org/10.1103/PhysRevLett.113.140501}{10.1103/PhysRevLett.113.140501}}.

\bibitem[WWS16]{wang16a}
Y.~Wang, X.~Wu, and V.~Scarani.
\newblock {All the self-testings of the singlet for two binary measurements}.
\newblock {\em New J. Phys.}, 18: 025021, 2016.
\newblock \\
  \texttt{DOI:\,\href{http://dx.doi.org/10.1088/1367-2630/18/2/025021}{10.1088/1367-2630/18/2/025021}}.

\bibitem[YN13]{yang13a}
T.~H. Yang and M.~Navascu{\'e}s.
\newblock {Robust self-testing of unknown quantum systems into any entangled
  two-qubit states}.
\newblock {\em Phys. Rev. A}, 87: 050102(R), 2013.
\newblock \\
  \texttt{DOI:\,\href{http://dx.doi.org/10.1103/PhysRevA.87.050102}{10.1103/PhysRevA.87.050102}}.

\end{thebibliography}
\newcommand{\etalchar}[1]{$^{#1}$}

\end{document}